%% file: asymmetric_erasures_v13.tex
\documentclass[11pt, letterpaper, onecolumn]{IEEEtran}

\input{preamble}

\usepackage{cite}

\begin{document}
\flushbottom

\title{Asymmetric Evaluations of Erasure and Undetected Error Probabilities } 
\author{Masahito Hayashi$^\dagger$, {\em Senior Member, IEEE}$\,\,\, $   $\,\,$ Vincent Y.~F.\ Tan$^\ddagger$, {\em Senior Member, IEEE} 
\thanks{$^\dagger$ M.~Hayashi is with the  Graduate School of Mathematics, Nagoya University, and the Center for Quantum Technologies (CQT),  National University of Singapore   (Email: {masahito@math.nagoya-u.ac.jp}).   }  \thanks{$\ddagger$ V.~.Y.~F. Tan is with the Department of Electrical and Computer Engineering and the Department of Mathematics, National University of Singapore (Email:  {vtan@nus.edu.sg}).   } \thanks{This paper was presented in part at the 2015 International Symposium on Information Theory in Hong Kong. } \thanks{
MH is partially supported by a MEXT Grant-in-Aid for Scientific Research (A) No.\ 23246071. 
MH is also partially supported by the National Institute of Information and Communication Technology (NICT), Japan.
The Centre for Quantum Technologies is funded by the Singapore Ministry of Education and the National Research Foundation as part of the Research Centres of Excellence programme.
VT's  research is supported by NUS   grants   R-263-000-A98-750/133.} \thanks{ Copyright (c) 2014 IEEE. Personal use of this material is permitted.  However, permission to use this material for any other purposes must be obtained from the IEEE by sending a request to pubs-permissions@ieee.org.}}


\maketitle

\begin{abstract}
The problem of channel coding with the erasure option is revisited for discrete memoryless channels. The interplay between the code rate, the undetected and total error probabilities is characterized. Using the information spectrum method, a sequence of codes of increasing blocklengths $n$ is designed to illustrate this tradeoff. Furthermore, for additive discrete memoryless channels with uniform input distribution, we establish that our analysis   is tight with respect to the ensemble average. This is done by analysing the ensemble performance in terms of a tradeoff between the code rate, the undetected and the  total errors. This tradeoff is parametrized by the threshold in a generalized likelihood ratio test. Two asymptotic regimes are studied. First, the code rate   tends to the capacity of the channel at a rate slower than $n^{-1/2}$ corresponding to the moderate deviations regime. In this case, both error probabilities decay subexponentially and asymmetrically. The precise decay rates are characterized. Second, the code rate  tends  to capacity at a rate of $n^{-1/2}$. In this case, the total error probability is asymptotically a positive constant while the undetected error probability decays as $\exp(- bn^{ 1/2})$ for some $b>0$. The proof techniques  involve applications of   a modified (or ``shifted'') version  of the G\"artner-Ellis theorem  and the type class enumerator method to characterize the asymptotic behavior of a sequence of cumulant generating functions. 
\end{abstract}


\section{Introduction}  \label{sec:intro}
\subsection{Background}
In channel coding, we are interested in designing a code that can reliably decode a message  sent through a noisy channel. However, when the effect of the noise is so large such that the  decoding system is not sufficiently confident of which message was sent, it is preferable to declare that an {\em erasure} event has occurred. In this way, the system avoids declaring that an incorrect message was sent,  a costly mistake, and may use an {\em automatic repeat request} (ARQ) protocol or {\em decision feedback} system to resend the intended message. This paper revisits the information-theoretic limits of channel coding with the erasure option. 

It has long been known since Forney's seminal paper on decoding with the erasure option and list decoding \cite{Forney68}  that the optimum decoder for a given codebook has the following structure: It  outputs the message   for which the likelihood of that message given the channel output exceeds a multiple $\exp(nT)$ (where $n$ is the blocklength of the code) of the sum of all the other likelihoods.  This is a generalization of the  likelihood ratio test which underlies  the Neyman-Pearson lemma for binary hypothesis testing.  For erasure decoding, the threshold $T$ is   set to a positive number so that the decoding regions are disjoint and furthermore, the erasure region is non-empty. Among our other contributions in this paper, we examine other possibly suboptimal decoding regions. 

If the threshold $T$ in Forney's decoding regions   is a fixed positive number not tending to zero, then it is known from his analysis  \cite{Forney68} and many follow-up works \cite{sgb, tel94, Blinovsky, Moulin09, merhav08,merhav13,merhav14, somekh11, merhav_FnT} that both the undetected error probability and the erasure probability decay exponentially fast in $n$ for an appropriately chosen codebook. Typically, and following in the spirit of Shannon's seminal work~\cite{Shannon48}, the codebook is randomly chosen. The constant $T$   serves to tradeoff between the two error probabilities.  This exponential decay in both error probabilities   corresponds to {\em large deviations} analysis. However, there is substantial motivation to study other asymptotic regimes to gain greater insights about the fundamental limits of  channel codes with the erasure option. This corresponds to setting the threshold $T$ to be a positive sequence that tends to zero as the  blocklength $n$ grows. 

Strassen~\cite{Strassen} pioneered the {\em fixed error probability} or {\em second-order asymptotic} analysis for  discrete memoryless channels (DMCs) without the erasure option.  There have been prominent works recently in this area by Hayashi~\cite{Hayashi09} and Polyanskiy, Poor and Verd\'u \cite{PPV10}. See \cite{TanBook} for a review.  Altu\u{g}  and Wagner~\cite{altug14b}  pioneered  the {\em moderate deviations} analysis for DMCs and Tan~\cite{Tan12} considered the rate-distortion counterpart for discrete and Gaussian sources. Second-order and moderate deviations analyses respectively correspond to operating at coding rates that have a deviation of  $\Theta(n^{-1/2})$ and  $\omega(n^{-1/2})$ from  the first-order fundamental limit, i.e., the capacity or the rate-distortion function.  Tan and Moulin~\cite{TanMoulin14} recently studied the information-theoretic limits of channel coding with erasures where both the undetected and total error probabilities are fixed at positive constants.

\subsection{Main Contributions} \label{sec:main_con}
In this work, we study different regimes for the errors and erasure problems. In particular, we analyze the  {\em moderate deviations}~\cite{PV10a, altug14b} and {\em mixed} regimes. For moderate deviations, the code rate tends towards capacity but deviates from it by a sequence that grows slower than $n^{-1/2}$. For the mixed regime, the undetected error is designed to decay as $\exp(- b n^{1/2})$ for some $b>0$, but the total error is asymptotically a positive constant governed by the Gaussian distribution. Our main contributions are detailed as follows.

First, for the achievability  results, we draw on ideas from information spectrum analysis \cite{Han10} to  present a sequence of  block  codes with the erasure option that demonstrate the above-mentioned asymmetric tradeoff between the undetected  and total error probabilities. 

Second, and equally importantly, we show that  our so-constructed codes above are {\em tight with respect to the ensemble average}, or more succinctly,  {\em ensemble-tight} for additive DMCs with the uniform random coding distribution.  This means that   our  ensemble evaluation of the two error probabilities   (averaged over the random codebook)  is tight in some asymptotic sense to be made precise in the statements.  To prove these statements, we consider Forney's decoding regions \cite{Forney68} where the threshold parameter $T$ depends on $n$ and, in particular, is set to be a decaying sequence $\Theta(n^{-t})$ where $t \in (0,1/2 ]$. We show that  both the undetected and total error probabilities decay subexponentially (i.e., the  moderate deviations regime~\cite{PV10a, Tan12,altug14b, AWK13,TWH14}) and asymmetrically in the sense that their decay rates are different. These decay rates depend on $t$ and also the implied constant the  $\Theta(n^{-t})$ notation. In fact, we characterize the precise tradeoff between these error probabilities, the code rate as well as the threshold. Our technique, which is based on the type class enumerator method~\cite{merhav08, merhav_FnT, merhav13,merhav14, somekh11},  carries over to the  mixed  regime  in  which the total error probability is asymptotically a constant \cite{Strassen,PPV10, Hayashi09} while the undetected error decays as $\exp(- b n^{1/2})$. Just as for the pure moderate deviations setting, we characterize the precise tradeoffs between the different parameters in the system. The decay rates turn out to be the same as for the achievability results showing that the   decoder designed based on information spectrum analysis is, in fact, asymptotically optimal, i.e., Forney's decoding regions (together with our analyses)  trade off the Pareto-optimal curve between the two error probabilities.

Finally, an auxiliary contribution of the present work is a new mathematical tool. We develop  a modified (``shifted'')  version of the G\"artner-Ellis theorem~\cite[Theorem 2.3.6]{Dembo}  to prove our results concerning the asymptotics of the undetected and total error probabilities under both the moderate and mixed regimes. This generalization, presented in Theorem \ref{thm7}, appears to be distinct from other generalizations of the G\"artner-Ellis theorem in the literature  (e.g., \cite{chen00, joutard}). It turns out to be very useful for our application and may be of independent interest in other information-theoretic settings. A self-contained proof containing some novel proof techniques  is contained in Appendix \ref{app:ge}. 

\subsection{Paper Organization}
This paper is organized as follows: In Section~\ref{sec:notation}, we state our notation and the problem setup precisely. The main results are detailed in Section~\ref{sec:results} where the direct results are in Section~\ref{sec:direct} and the ensemble converse results in Section~\ref{sec:ens_results}. The proofs of the main results are deferred to Section~\ref{sec:prfs}. We conclude our discussion and suggest avenues for future work in Section \ref{sec:concl}. The appendices contain some auxiliary mathematical tools including the modification  of the G\"artner-Ellis theorem for general orders, which we use to estimate the both  errors. This is presented as Theorem~\ref{thm7} in   Appendix \ref{app:ge}.

\section{Notation and Problem Setting}  \label{sec:notation}

\subsection{Notation}
In this paper, we adopt standard notation in information theory, particularly in the book by Csisz\'ar and K\"orner~\cite{Csi97}. Random variables are denoted by upper case (e.g., $X$) and their realizations by lower case (e.g., $x$). All alphabets of the random variables are finite sets and are denoted by calligraphic font (e.g., $\calX$). A sequence of letters from the $n$-fold Cartesian product $\calX^n$ is denoted by boldface $\bx = (x_1,\ldots, x_n)$.  A sequence of random variables is denoted using a  superscript, i.e., $X^n = (X_1,\ldots, X_n)$. Information-theoretic quantities are denoted in the usual way, e.g.,  $H(P)$ is the entropy of the random variable $X$ with distribution $P$.   The set of all probability mass functions on a finite set $\calX$ is denoted by $\calP(\calX)$ while the subset of types (empirical distributions) with denominator $n$ is denoted as $\calP_n(\calX)$. The set of all sequences with type $P \in\calP_n(\calX)$, the {\em type class}, is denoted as $\calT_P = \{\bx=(x_1,\ldots, x_n)\in\calX^n:\sum_{i=1}^n\bone\{ x_i=a\}=nP(a),\forall\, a\in\calX\}$. The $\ell_1$ (twice the variational) distance between $P,Q\in\calP(\calX)$ is denoted as $\|P-Q\|_1=\sum_{x \in\calX} |P(x)-Q(x)|$.  All logs and exps are with respect to the natural base $\rme$.

\subsection{Discrete Memoryless Channels (DMCs)}
We consider a DMC  $W$ with input alphabet $\calX$ and output alphabet $\calY$. This is denoted as $W:\calX\to\calY$. By memoryless (and stationary), this means that given a sequence of input letters $\bx =(x_1,\ldots, x_n)\in\calX^n$ the probability of the  output letters $\by  = (y_1,\ldots, y_n)\in\calY^n$ is the product $\prod_{i=1}^n W(y_i|x_i)$.   The capacity of the DMC is denoted as 
\begin{equation}
C=C(W):=\max\{ I(P_X,W) : P_X\in\calP(\calX)\}.
\end{equation}
Let the set of capacity-achieving input distributions be 
\begin{equation}
\Pi=\Pi(W) := \{P_X\in\calP(\calX): I(P_X,W) = C(W)\}.
\end{equation}
This set is  compact. 
\subsection{Additive DMCs} \label{sec:additive}
A DMC is called  {\em additive} if  $\calX=\calY=\{0,1,\ldots, d-1\}$  for some $d\in\bbN$ and there exists a probability mass function $P \in\calP(\calX)$ with positive entries $P(x) >0,x\in\calX$ such that 
\begin{equation}
W(y|x) = P(y-x) \label{eqn:additive}
\end{equation}
where the $-$ in \eqref{eqn:additive} is understood to be modulo $d$, i.e., the subtraction operation in the {\em additive group} $(\{0,1,\ldots, d-1\},+)$. In other words, $Y=X+Z$ (mod $d$) where the noise $Z$ has distribution $P$. Consequently, $P$ is also called the  {\em noise distribution}. The   capacity of the additive channel $W$ is $C = \log d - H(P)$ and which is achieved (possibly non-uniquely) by  the uniform distribution on  $\{ 0,1,\ldots, d-1\}$~\cite[Theorem 7.2.1]{Cov06}. This class of channels, while somewhat restrictive, includes important DMCs such as the binary symmetric channel (BSC)  where $d = 2$ and $P(0)=q$ and $P(1) = 1-q$ and $q \in (0,1)$ is the crossover probability. Also, additive DMCs simplify analyses in other problems in Shannon theory such as in the error exponent analysis of the performance of linear codes~\cite{Csi82}.
\subsection{Channel Coding with the Erasure Option}
We consider a channel coding problem in which a message taking values in $\{1,\ldots, {M_n}\}$ uniformly at random is to be transmitted across a noisy channel $W^n$. An {\em encoder} $\rvf :\{1,\ldots, {M_n}\}\to\calX^n$ transforms the message to a codeword. The {\em codebook} $\calC_n=\{\bx_1,\ldots, \bx_{M_n}\}$ where $\bx_m=\rvf(m)$ is the set of all codewords.  The channel $W^n$ then applies a random transformation to the chosen codeword $\bx_m\in\calX^n$ resulting in $\by\in\calY^n$. A {\em decoder}  $\rvd : \calY^n\to\{0,1,\ldots, {M_n}\}$ either declares an estimate of the message or outputs an erasure symbol, denoted as $0$. The decoding operation can thus be regarded as partition of the output space $\calY^n$ into ${M_n}+1$ disjoint {\em decoding regions} $\calD_0,\calD_1,\ldots,\calD_{M_n} \subset\calY^n$, where $\calD_m := \rvd^{-1}(m)$. The set of all $\by\in\calD_0$ leads to an {\em erasure} event. 

\subsection{Total and Undetected Error Probabilities}
Given a codebook $\calC_n$, one can define two undesired error events for $n$ uses of the DMC. The first is the event in which the decoder does not make the correct decision, i.e., if message $m$ is sent, it declares either an erasure $0$ or outputs an incorrect message $m'\ne m$ (more precisely, $m\in\{1,\ldots, M_n\}\setminus \{m\}$). The probability of this event $\calE_1$ can be written as 
\begin{equation}
\Pr(\calE_1 | \calC_n) = \frac{1}{{M_n}}\sum_{m=1}^{M_n} \sum_{\by\in\calD_m^c} W^n(\by|\bx_m). \label{eqn:p1}
\end{equation}
This is the {\em total error probability}.

The other error event is $\calE_2$, which is defined as the event of declaring an incorrect message, i.e., if $m$ is sent, the decoder declares that $m'  \ne m$ is sent instead. This {\em undetected error  probability} can be written as 
\begin{equation}
\Pr(\calE_2 | \calC_n) = 
\frac{1}{{M_n}}\sum_{m=1}^{M_n} \sum_{\by\in\calD_m } \sum_{m'  \ne m  } W^n(\by|\bx_{m'}). \label{eqn:p2}
\end{equation}
One usually designs the codebook $\calC_n$ and the decoder $\rvd$ such that $\Pr(\calE_2 | \calC_n)$ is much smaller than $\Pr(\calE_1 | \calC_n)$, because undetected errors are usually more undesirable than erasures. 

\section{Main results} \label{sec:results}
\subsection{Direct Results} \label{sec:direct} 
We now state our main result in this paper concerning the asymmetric evaluation of   $\Pr(\calE_1 | \calC_n)$  and $\Pr(\calE_2 | \calC_n)$ which correspond to the total error probability and the undetected error probability respectively. Define the conditional information variance of an input distribution $P_X$ and the channel $W$ as 
\begin{equation}
V(P_X,W) :=\sum_{x\in\calX}P_X(x) \sum_{y\in\calY} W(y|x) \left[ \log\frac{W(y|x)}{P_XW(y)}- D(W(\cdot|x) \| P_XW) \right]^2,
\end{equation}
where $P_X W(y) = \sum_x P_X(x) W(y|x)$ is the output distribution  induced by $P_X$ and $W$. This quantity is finite whenever $W(\cdot|x)\ll P_XW $ for all $x$.    We further define the minimum and maximum conditional information variances as 
\begin{align}
V_{\max}(W)& :=\max_{P_X \in\Pi } V(P_X,W) 
\label{6-6-1} \quad\mbox{and} \\
V_{\min}(W)&:=\min_{P_X\in\Pi   } V(P_X,W) . 
\label{6-6-2}
\end{align}
Since $\Pi$ is compact and $P_X\mapsto V(P_X,W)$ is continuous, there exists capacity-achieving input distributions $P_X \in\Pi$ that achieve  both  $V_{\min}(W)$ and $V_{\max}(W)$ and so they are finite. The $P_X$ that achieves $V_{\min}(W)$ may not be the same as that achieving $V_{\max}(W)$.
Note that for all $P_X \in\Pi$, we have  $V(P_X,W) = U(P_X, W)$ \cite[Lem.~62]{PPV10}, where the unconditional information variance $U(P_X, W)$ is defined as 
\begin{equation}
U(P_X,W) :=\sum_{x\in\calX}P_X(x) \sum_{y\in\calY} W(y|x) \left[ \log\frac{W(y|x)}{P_XW(y)}- C \right]^2. \label{eqn:def_U}
\end{equation}
We assume that the channel $W$ satisfies $V_{\min}(W)>0$ throughout. This holds for all interesting DMCs (except some degenerate cases) and we make this assumption which is standard in moderate deviations analysis \cite{PV10a, altug14b}. If $V_{\min}(W)=0$, the conclusion from the moderate deviations theorem~\cite[Theorem 3.7.1]{Dembo} fails to hold. 
\begin{theorem}[Moderate Deviations Regime Direct] \label{thm1}
Let $0< t < 1/2$ and $a>b>0$. Set the number of codewords\footnote{We ignore integer constraints on the number of codewords $M_n$. We simply set $M_n$ to the nearest integer to the number satisfying~\eqref{eqn:sizeM}. } $M_n$ to satisfy 
\begin{equation}
\log M_n=  nC - a n^{1-t}. \label{eqn:sizeM}
\end{equation}
There exists a sequence of codebooks $\calC_n$ with $M_n$ codewords 
such that
the two error probabilities satisfy 
\begin{align}
\lim_{n\to\infty}\,\, -\frac{1}{n^{1- 2t}} \log
\Pr(\calE_1 | \calC_n) & = \frac{(a-b)^2}{2 V_{\min}(W) }
,\quad\mbox{and}\label{eqn:e1_md_direct} \\
\liminf_{n\to\infty}\,\, -\frac{1}{n^{1- t}} \log
\Pr(\calE_2 | \calC_n) & \ge b .
\label{eqn:e2_md_direct} 
\end{align}
\end{theorem}
The proof of this result can be found in Section \ref{sec:prf1}. We assume that $a>b$ because if we demand that the undetected error probability decays as in  \eqref{eqn:e2_md_direct}, we must have that the rate $\frac{1}{n}\log M_n$ backs off further from capacity per \eqref{eqn:sizeM}.  Furthermore, $b<0$ corresponds to the list region which we do not discuss in detail in this paper.

 Interestingly, we do not analyze the optimal decoding regions prescribed by Forney~\cite{Forney68} and described in \eqref{eqn:decision_reg} in the sequel. We consider the following   regions $\{\tilde{\calD}_m\}_{m=1}^{M_n}$ motivated by information spectrum analysis~\cite{Han10}: 
\begin{equation}
\tilde{\calD}_m :=
\left\{ \by: \log
 \frac{W^n(\by|\bx_m)}{(P_XW)^n(\by) }
\ge 
 \log M_n   +  b n^{1-t}  
\right\} ,  \label{eqn:decision_reg1}
\end{equation}
where $P_X$ is a capacity-achieving input distribution. We choose $P_X$ to achieve either $V_{\min}(W)$ or $V_{\max}(W)$ in the proofs. Now we define the  set of all $\by\in\calY^n$ that leads to an erasure event in terms of $\{\tilde{\calD}_m\}_{m=1}^{M_n}$ as 
\begin{equation}
\hat{\calD}_0 := \bigg( \bigcap_{m=1}^{M_n} \tilde{\calD}_m^c \bigg) \cup \bigg( \bigcup_{m\ne m'} (\tilde{\calD}_m\cap \tilde{\calD}_{m'}) \bigg). \label{eqn:D0}
\end{equation}
Then, the {\em decoding region for message $m=1,\ldots, M$} is defined to be 
\begin{equation}
\hat{\calD}_m := \tilde{\calD}_m \setminus \hat{\calD}_0 . 
 \label{eqn:Dm}
\end{equation}
The erasure region is $\hat{\calD}_0$ described in \eqref{eqn:D0}. A moment's of thought reveals that $\hat{\calD}_0,\hat{\calD}_1,\ldots, \hat{\calD}_{M_n}$ are mutually disjoint and furthermore $\cup_{m=0}^{M_n}\hat{\calD}_m =\calY^n$.  The intuition behind the decoding regions in \eqref{eqn:D0}--\eqref{eqn:Dm} is as follows: The erasure region in \eqref{eqn:D0} is the union of all the complements of nominal regions  $\tilde{\calD}_m^c$ and the sets of pairwise intersections which potentially cause confusion in decoding, namely,  $\tilde{\calD}_m\cap \tilde{\calD}_{m'}$. After defining  the erasure region $\hat{\calD}_0 $, we remove this from the nominal regions  $\tilde{\calD}_m$  to form the actual decoding region for each message $\hat{\calD}_m$.    Note that in the ensemble tightness results to be presented in Section~\ref{sec:ens_results} we do not analyze the information spectrum decoding regions in  \eqref{eqn:decision_reg1}--\eqref{eqn:Dm}. Rather we analyze the {\em optimal  decoder} suggested by Forney~\cite{Forney68}. Hence, the decoding regions in  \eqref{eqn:decision_reg1}--\eqref{eqn:Dm}, in general,  may not be asymptotically optimal, unlike Forney's  decoding regions. However, we do show that these decoders are asymptotically optimal for additive DMCs.

Theorem \ref{thm1} corresponds to the so-called moderate deviations regime  in channel coding considered by Altu\u{g} and Wagner~\cite{altug14b} and Polyanskiy and Verd\'u~\cite{PV10a}. 
Thus, the appearance of the   term $V_{\min}(W)$ in the results is natural.  However, notice that the error probabilities $\Pr(\calE_1 | \calC_n)$ and $\Pr(\calE_2 | \calC_n)$ decay asymmetrically. By that, we mean that the rates of decay are different---$\Pr(\calE_1 | \calC_n)$ decays as $\exp(-\Theta(n^{1-2t}))$ while $\Pr(\calE_2 | \calC_n)$ decays as $\exp(-\Omega(n^{1-t}))$.
 
When $t=1/2$, 
we observe different asymptotic scaling from that in Theorem~\ref{thm:md}. 
Define 
\begin{equation}
 \varphi(w) := \frac{1}{\sqrt{2\pi}}\exp\left(-\frac{w^2}{2}\right) \label{eqn:pdf}
  \end{equation} 
  to be the probability density function of a standard Gaussian annd 
\begin{equation}
\Phi(\alpha) := \int_{-\infty}^\alpha \varphi(w)\,  \rmd w
\end{equation}
 to be the cumulative distribution function of a standard Gaussian.

\begin{theorem}[Mixed Regime Direct] \label{thm2}
Let $b>0$, $a \in \bbR$,
and $M_n$ chosen as in \eqref{eqn:sizeM} with $t=1/2$. 
There exists a sequence of codebooks $\calC_n$ with $M_n$ codewords 
such that
$\Pr(\calE_2 | \calC_n)$ satisfies 
\begin{align}
\lim_{n\to\infty}\,\,  \Pr(\calE_1 | \calC_n) 
&=
\left\{
\begin{array}{ll}
   \Phi\Big( \frac{b-a}{\sqrt{V_{\max}(W)}}\Big) & \hbox{ if } a \le 0 \\
   \Phi\Big( \frac{b-a}{\sqrt{V_{\min}(W)}}\Big) & \hbox{ if } a > 0.
\end{array}
\right.
,\quad\mbox{and}
\label{eqn:e1_clt_direct} \\
\liminf_{n\to\infty}\,\, -\frac{1}{\sqrt{n}} \log
\Pr(\calE_2 | \calC_n) & \ge b .
\label{eqn:e2_clt_direct} 
\end{align}
\end{theorem}
The proof of this result can be found in Section \ref{sec:prf2}. Observe that the first error  probability is in the central limit regime \cite{Strassen, PPV10, Hayashi09} while the second scales as $\exp(- \sqrt{n}  \, b)$, which is in the moderate deviations regime \cite{altug14b,PV10a}. Thus, we call this the {\em mixed regime}.

\subsection{Tightness With Respect to the Ensemble Average}  \label{sec:ens_results}
It is, at this point, not clear that the codes we proposed in Section \ref{sec:direct} are asymptotically optimal. In this section, we demonstrate the tightness of the decoder for {\em additive}  DMCs with uniform input distribution (which is a capacity-achieving input distribution for additive DMCs).  We consider an ensemble evaluation of the two error probabilities. That is, we evaluate the probabilities of total and undetected errors averaged over the random code and show that this evaluation is tight in some asymptotic sense to be made precise in the statements. For brevity, we also call this evaluation {\em ensemble tightness} or {\em ensemble converse}.
Similarly to \eqref{eqn:sizeM},  the sizes  of the codes    we consider $\{M_n\}_{n\in\bbN}$  take the form 
\begin{equation}
\log {M_n} = n C - a n^{1-t}  \label{eqn:sizeM2}
\end{equation}
where $C = \log d - H(P)$ is the capacity of the additive channel and  $0<t\le 1/2$. When $t<1/2$ (resp.\ $t=1/2$), the code size is in the moderate deviations (resp.\ central limit or mixed) regime.

We now state our main results in this paper concerning the asymmetric evaluation of   $\Pr(\calE_1 | \calC_n)$  and $\Pr(\calE_2 | \calC_n)$  corresponding to the total error probability and the undetected error probability respectively. We define the {\em varentropy}~\cite{verdu14} or {\em source dispersion}~\cite{kost12} of the additive noise $P$ as 
\begin{equation}
V(P):= \sum_{z=0}^{d-1} P(z) \left[ \log\frac{1}{P(z)}-H(P)\right]^2.
\end{equation}
This is simply the variance of the self-information random variable $-\log P(Z)$ where $Z$ is distributed as $P$.  We assume that $V(P)>0$ throughout. It is easy to see that because of the additivity of the channel, the $\epsilon$-dispersion~\cite{PPV10}  of $W$  is $V(P)$ for every $\epsilon\in (0,1)$, i.e., $V_{\min}(W)=V_{\max}(W)=V(P)$. 

In the following, we emphasize that the uniform distribution will be chosen as the input distribution of the code. This is equivalent to choosing the $M_n$ codewords where  each codeword is drawn uniformly at random from $\{0,1,\ldots, d-1\}^n$.

\begin{theorem}[Moderate Deviations Regime   Converse] \label{thm:md}
Let $0<t<1/2$ and $a> b > 0$. Consider a sequence of random codebooks $\calC_n$  with $M_n$ codewords where each codeword is drawn uniformly at random from $\{0,1,\ldots, d-1\}^n$ and $M_n$ satisfies \eqref{eqn:sizeM2}.  Let $W$ be an additive DMC. 
When the expectation of the total error satisfies 
\begin{align}
  \liminf_{n\to\infty}\,\, -\frac{1}{n^{1-2t}} \log\bbE_{\calC_n}\big[\Pr(\calE_1 | \calC_n) \big]& \ge
    \frac{(a-b)^2}{2 V(P) } ,\label{H1}
\end{align}
then the expectation of the undetected error satisfies 
\begin{align}
 \limsup_{n\to\infty}\,\, -\frac{1}{n^{1- t}} \log\bbE_{\calC_n}\big[\Pr(\calE_2 | \calC_n) \big] & \le   b.\label{H2}
\end{align}
Conversely, when the expectation of the undetected error satisfies 
\begin{align}
 \liminf_{n\to\infty}\,\, -\frac{1}{n^{1- t}} \log\bbE_{\calC_n}\big[\Pr(\calE_2 | \calC_n) \big] & \ge   b,\label{H3}
\end{align}
then the expectation of the total error satisfies 
\begin{align}
  \limsup_{n\to\infty}\,\, -\frac{1}{n^{1-2t}} \log\bbE_{\calC_n}\big[\Pr(\calE_1 | \calC_n) \big]& \le
    \frac{(a-b)^2}{2 V(P) } .\label{H4}
\end{align}
\end{theorem}

\begin{theorem}[Mixed Regime    Converse] \label{thm:clt}
Let $b > 0$, $a\in\bbR$ and $M_n$ chosen according to \eqref{eqn:sizeM2} with $t=1/2$.    Consider a sequence of random codebooks $\calC_n$  with $M_n$ codewords where each codeword is drawn uniformly at random from $\{0,1,\ldots, d-1\}^n$, the decoding regions are chosen according to \eqref{eqn:decision_reg} with thresholds
\eqref{eqn:thres_md}.  Let $W$ be an additive DMC. 
When the expectation of the total error satisfies 
\begin{align}
\limsup_{n\to\infty}\,\,  \bbE_{\calC_n} \big[\Pr(\calE_1 | \calC_n)\big] & \le
\Phi\bigg( \frac{b-a}{\sqrt{V(P)}}\bigg)  
\label{H5}
\end{align}
then the expectation of the undetected error satisfies 
\begin{align}
 \limsup_{n\to\infty}\,\, -\frac{1}{\sqrt{n}} \log\bbE_{\calC_n}\big[\Pr(\calE_2 | \calC_n) \big] & \le   b.
\label{H6}
\end{align}
Conversely, when the expectation of the undetected error satisfies 
\begin{align}
 \liminf_{n\to\infty}\,\, -\frac{1}{\sqrt{n}} \log\bbE_{\calC_n}\big[\Pr(\calE_2 | \calC_n) \big] & \ge   b,
\label{H7}
\end{align}
then the expectation of the total error satisfies 
\begin{align}
\liminf_{n\to\infty}\,\,  \bbE_{\calC_n} \big[\Pr(\calE_1 | \calC_n)\big] & 
\ge\Phi\bigg( \frac{b-a}{\sqrt{V(P)}}\bigg)  .
\label{H8}
\end{align}
\end{theorem}

These theorems imply that 
if we generate our encoder according to the uniform distribution
even if we improve our decoder,
we cannot improve both errors.
That is, these theorems show the asymptotic optimality of our codes for the additive
channel. The proofs of these theorems follow immediately from Lemmas \ref{lem:md} and \ref{lem:clt} to follow.

To prove these theorems we need to develop Lemmas \ref{lem:md} and \ref{lem:clt} in the following. We recall Forney's  result in~\cite{Forney68}
 that for a given codebook $\calC_n:=\{\bx_1,\ldots, \bx_{M_n}\}$, the Pareto-optimal decoding region  for each  message $m  \in \{1,\ldots, {M_n}\}$ is given by 
\begin{equation}
\calD_m :=\left\{ \by:\frac{W^n(\by|\bx_m)}{\sum_{m'\ne m}W^n(\by|\bx_{m'})}\ge \exp(n{T_n})\right\},  \label{eqn:decision_reg}
\end{equation}
where ${T_n} > 0$  is a threshold parameter that serves to trade off between the two error probabilities $\Pr(\calE_1 | \calC_n)$  and $\Pr(\calE_2 | \calC_n)$. This is a  generalization of the Neyman-Pearson lemma. Because ${T_n} > 0$, the regions are disjoint.  We let $\calD_0$ denote the set of all $\by$ that leads to an erasure, i.e., 
\begin{equation}
\calD_0 := \calY^n\setminus \bigsqcup_{m=1}^{M_n}\calD_m.
\end{equation}

In the literature on decoding with an erasure option (e.g.,~\cite{Forney68,sgb, tel94, Blinovsky, Moulin09, merhav08,merhav13,merhav14,  somekh11}), ${T_n}$ is usually kept at a constant (not depending on $n$), leading to results concerning tradeoffs between the {\em exponential} decay rates of $\Pr(\calE_1 | \calC_n)$  and $\Pr(\calE_2 | \calC_n)$, i.e., the {\em error exponents} of the total and undetected error probabilities. Our treatment is different. We let ${T_n}$ in the definitions of the decision regions $\calD_m$ in \eqref{eqn:decision_reg} depend on $n$ and show that the error probabilities $\Pr(\calE_1 | \calC_n)$  and $\Pr(\calE_2 | \calC_n)$ decay {\em subexponentially} and in an {\em asymmetric} manner, i.e.,  at different speeds.

\begin{lemma}[Moderate Deviations Regime Ensemble] \label{lem:md}
Let $0<  t<1/2$ and $a > b > 0$. Consider a sequence of random codebooks $\calC_n$  with $M_n$ codewords where each codeword is drawn uniformly at random from $\{0,1,\ldots, d-1\}^n$ and $M_n$ satisfies \eqref{eqn:sizeM2}. 
Let  the decoding regions be chosen as in \eqref{eqn:decision_reg} with  thresholds
\begin{equation}
{T_n} :=\frac{b}{n^{ t}}, \label{eqn:thres_md}
\end{equation}
 Let $W$ be an additive DMC.  Then, the expectation of the two error probabilities satisfy 
\begin{align}
  \lim_{n\to\infty}- \frac{1}{n^{1-2t}} \log\bbE_{\calC_n}\big[\Pr(\calE_1 | \calC_n) \big]
&  =\frac{(a-b)^2}{2 V(P) } 
  ,\quad\mbox{and}\label{eqn:e1_md} \\
  b n^{1- t}  +\frac{(a-b)^2}{2V(P)} n^{1-2t}  + o(n^{1-2t})& \le -\log\bbE_{\calC_n}\big[\Pr(\calE_2 | \calC_n) \big]  \le  b n^{1- t}  + o(n^{1- t})
\label{eqn:e2_md} 
\end{align}
\end{lemma}
The proof of this lemma is provided in Section~\ref{sec:prf-md}.  From this lemma, we can show Theorem \ref{thm:md} by a simple argument which we defer to Section~\ref{sec:optimality_prf}.   At this point, a few other comments concerning are in order.


This result again corresponds to the so-called moderate deviations regime  in channel coding considered by Altu\u{g} and Wagner~\cite{altug14b} and Polyanskiy and Verd\'u~\cite{PV10a}. Thus, the appearance of the varentropy term $V(P)$ in the results is natural.  
The total and undetected error probabilities in \eqref{eqn:e1_md} and \eqref{eqn:e2_md} can  be written as 
\begin{align}
\bbE_{\calC_n}\big[\Pr(\calE_1 | \calC_n ) \big] 
 & =  \exp\left(-  \frac{(a-b)^2}{2 V(P)}  \, n^{1-2t}     + o( n^{1-2t}) 
\right), 
\quad\mbox{and} \label{eqn:e1_md2} \\
\bbE_{\calC_n}\big[\Pr(\calE_2 | \calC_n ) \big] 
& =  \exp\big(-bn^{1-t}    + o(n^{1- t}) \big).
\label{eqn:e2_md2}
\end{align}
respectively. This  scaling is also different from those found in the literature which primarily focus on exponentially decaying probabilities~\cite{Forney68,sgb, tel94, Blinovsky, Moulin09, merhav08,merhav13,merhav14,  somekh11}  or non-vanishing error probabilities~\cite{TanMoulin14}. Both our total and undetected error  probabilities are designed to decay subexponentially fast in the blocklength $n$. Our proof technique involves estimating appropriately-defined cumulant generating functions and invoking a modified version of the G\"artner-Ellis theorem~\cite[Theorem 2.3.6]{Dembo}. The statement of this modified form of the G\"artner-Ellis theorem is presented as  Theorem~\ref{thm7} in Appendix~\ref{app:ge} and we provide a self-contained proof therein.  Similarly to the work by Somekh-Baruch and Merhav \cite{somekh11}, the two probabilities in \eqref{eqn:e1_md}--\eqref{eqn:e2_md} are asymptotic {\em equalities} (if we consider the normalizations $n^{1-2t}$ and $n^{1-t}$) rather than inequalities (cf.~\cite{Forney68,merhav08}). In fact for the lower  bound in  \eqref{eqn:e2_md}, we can even calculate  a higher-order asymptotic term scaling as $n^{1-2t}$ (but unfortunately, we do not yet have a matching upper bound for the higher-order term).

Next,   observe that  the undetected error decays  faster than the total error  because the former is  more undesirable than an erasure. If $a$  is increased for fixed $b$, the effective number of codewords is decreased so commensurately, the total error probability   $\Pr(\calE_1 | \calC_n )$ is also reduced.  Also,   if $b$  is increased (tending towards  $a$ from below), the probability of an erasure increases and so the probability of an undetected error decreases. This is evident in \eqref{eqn:e1_md2} where the coefficient $\frac{(a-b)^2 }{ 2V(P)}$ decreases and  in~\eqref{eqn:e2_md2} where the leading coefficient $b$     increases.  Thus, we observe a  delicate interplay between $a$  governing the code size and $b$, the parameter in the threshold.

Finally, if $T_n$ is negative (a case not allowed by Lemma~\ref{lem:md}). This corresponds to {\em list decoding} \cite{Forney68} where the decoder is allowed to output more than one message (i.e., a {\em list} of messages) and an error event occurs if and only if the transmitted message is not in the list. In this case,  $\Pr(\calE_2|\calC_n)$ no longer corresponds to the probability of undetected error. Rather, the expression for $\Pr(\calE_2|\calC_n)$ in \eqref{eqn:p2}  corresponds to the average number of incorrect codewords in the list corresponding to the  {\em overlapping} (non-disjoint) decision regions $\{\calD_m\}_{m=1}^{M_n}$.


\begin{lemma}[Mixed Regime Ensemble] \label{lem:clt}
Let   $b > 0$, $a\in\bbR$ and $M_n$ chosen according to \eqref{eqn:sizeM2} with $t=1/2$.  Consider a sequence of random codebooks $\calC_n$  with $M_n$ codewords where each codeword is drawn uniformly at random from $\{0,1,\ldots, d-1\}^n$, the decoding regions are chosen according to \eqref{eqn:decision_reg} with thresholds
\eqref{eqn:thres_md}.  Let $W$ be an additive DMC. 
Then, the  expectation of the two error probabilities satisfy 
\begin{align}
\lim_{n\to\infty}\,\,  \bbE_{\calC_n} \big[\Pr(\calE_1 | \calC_n)\big] 
& =\Phi\bigg( \frac{b-a}{\sqrt{V(P)}}\bigg)  
\quad\mbox{and} \label{eqn:e1_clt}  \\
   b \sqrt{n}+ \frac{(a-b)^2}{2V(P)} + o(1) & \le -   \log\bbE_{\calC_n}\big[\Pr(\calE_2 | \calC_n) \big] 
  \le  b \sqrt{n}  +  o(\sqrt{n})  . 
\label{eqn:e2_clt} 
\end{align}
\end{lemma}

The proof of this lemma is provided in Section~\ref{sec:prf-clt}. It is largely  similar to that for Lemma~\ref{lem:md} but for  the total error probability in~\eqref{eqn:e1_clt},  instead of invoking the G\"artner-Ellis theorem~\cite[Theorem 2.3.6]{Dembo}, we use the fact that if the cumulant generating function of a sequence of random variables $\{K_n\}_{n\in\bbN}$ converges to a quadratic function, $\{K_n\}_{n\in\bbN}$  converges in distribution  to a Gaussian random variable. However, this is not completely straightforward as we can only prove that the cumulant generating function converges pointwise for {\em positive} parameters (cf.~Lemma~\ref{lem:6}). We thus need to invoke a result by Mukherjea {\em et al.} \cite[Thm.~2]{muk06}  (building on initial work by Curtiss~\cite{curtiss}) to assert weak convergence.  (See Lemma \ref{l625} in Appendix \ref{app:conv}.) The asymptotic bounds in~\eqref{eqn:e2_clt}  are proved using  a modified version of the  G\"artner-Ellis theorem.


Here, ignoring the constant term in the lower bound, the undetected error probability  in \eqref{eqn:e2_clt} decays   as 
\begin{equation}
\bbE_{\calC_n}\big[\Pr(\calE_2 | \calC_n ) \big] =   \exp\left(- b\sqrt{n}  + o(\sqrt{n}  )\right). \label{eqn:prefact}
\end{equation}
The total (and hence, erasure)  error probability in \eqref{eqn:e1_clt} is asymptotically a constant depending on the varentropy of the noise distribution $P$, the threshold parametrized by $b$ and the code size parametrized by $a$.   Similarly to Lemma~\ref{lem:md}, if $b$ increases for fixed $a$, the likelihood of an erasure event occurring also increases but this decreases the undetected error probability as evidenced by \eqref{eqn:e2_clt}. The situation in which $b\downarrow  0$  for fixed $a$ recovers a special case of a recent result by Tan and Moulin~\cite[Thm.~1]{TanMoulin14} where the total error probability is kept constant at a positive constant and the undetected error probability vanishes. Note that for this result, we do not require that $a>b$ unlike what we assumed for the pure moderate deviations setting of Lemma~\ref{lem:md}. 

In the same way as we can show Theorem \ref{thm:md} from Lemma \ref{lem:md}, we can also use the exact same argument to show Theorem \ref{thm:clt} from Lemma \ref{lem:clt}. Thus, we omit the details here.

\section{Proofs of the Main Results}  \label{sec:prfs}
\subsection{Proof of Theorem~\ref{thm1}}\label{sec:prf1}
Choose any input distribution $P_X \in\Pi(W)$  achieving $V_{\min}(W)$ 
in (\ref{6-6-2}).
We consider choosing each codeword $\bx_m, m \in \{1,\ldots, M_n\}$ 
with the product distribution $P_X^n \in\calP( \calX^n)$. 
The expectation over this random choice of codebook is denoted as $\bbE_{\calC_n}[ \cdot]$. Now, we first consider $\Pr(\calE_1 |\calC_n)$.  Define the (capacity-achieving) output distribution $P_Y := P_X W$ and its $n$-fold memoryless extension $P_Y^n$. 
Next, we consider regions  
$\tilde{\cal D}_m$   defined in \eqref{eqn:decision_reg1}. The expectation over the code of the $W^n(\cdot|X_{m'}^n)$-probability of $\tilde{\cal D}_m$  can be evaluated as
\begin{align}
& \bbE_{\calC_n} \left[ 
\sum_{\by} W^n(\by|X^n_{m'})
\bone
\left\{ 
W^n(\by|X^n_m)
\ge 
M_n \exp( b n^{1-t})
P_Y^n (\by)
\right\}
\right]   \nn\\
&= 
\bbE_{X_m^n} 
\left[
\sum_{\by}
P_Y^n(\by)
\bone
\left\{ 
W^n(\by|X^n_m)
\ge 
M_n \exp( b n^{1-t})
P_Y^n (\by)
\right\}
\right]  \label{eqn:explain_PY} \\
&\le 
\bbE_{X_m^n} 
\left[
\sum_{\by}
M_n^{-1} \exp(- b n^{1-t})
W^n(\by|X^n_m)
\bone
\left\{ 
W^n(\by|X^n_m)
\ge 
M_n \exp( b n^{1-t})
P_Y^n (\by)
\right\}
\right]   \\
&\le  M_n^{-1} \exp(- b n^{1-t})
\end{align}
for $m'\neq m$, where \eqref{eqn:explain_PY} is because of independence of codeword generation and  $\bbE_{X^n_{m'}}[W^n(\by|X^n_{m'})]=P_{Y }^n(\by)$. 
Since $ \hat{\calD}_m \subset \tilde{\calD}_m$, by the definition of $\tilde{\calD}_m$ in~\eqref{eqn:decision_reg1},  
the expectation of 
the undetected error probability over the random codebook can be written as 
\begin{align}
 \bbE_{\calC_n} \big[ \Pr(\calE_2|\calC_n) \big]  
& \le
\bbE_{\calC_n} \left[ 
\frac{1}{{M_n}}\sum_{m=1}^{M_n} \sum_{\by  } \sum_{m'\ne m} W^n(\by|X^n_{m'})
\bone
\left\{ 
W^n(\by|X^n_m)
\ge 
M_n \exp( b n^{1-t})
P_Y^n (\by)
\right\}
\right]\\
& \le
\frac{1}{{M_n}}\sum_{m=1}^{M_n} \sum_{m'\ne m}
M_n^{-1} \exp(- b n^{1-t})\\
&=\frac{M_n-1}{M_n}\exp(- b n^{1-t}) \\
&\le \exp(- b n^{1-t}).\label{6-6-4}
\end{align}
Hence, this bound verifies \eqref{eqn:e2_md_direct}.  

By the definition of $\hat{\calD}_m$ for $m = 0,1 , \ldots, M_n$ in \eqref{eqn:decision_reg1} and~\eqref{eqn:D0}, we know that 
\begin{equation}
\hat{\calD}_m^c = \tilde{\calD}_m^c\cup\hat{\calD}_0=\tilde{\calD}_m^c \cup\bigcup_{m'\ne m}  \left(\tilde{\calD}_m\cap \tilde{\calD}_{m'} \right)\subset \tilde{\calD}_m^c \cup\bigcup_{m'\ne m}    \tilde{\calD}_{m'} . \label{eqn:sets}
\end{equation}
The expectation of the total error probability over the random codebook can be written as 
\begin{align}
&\bbE_{\calC_n} \big[ \Pr(\calE_1|\calC_n) \big]  \nn\\*
 &=\bbE_{\calC_n} \Bigg[ \frac{1}{M_n}\sum_{m=1}^{M_n}\sum_{\by\in\hat{\calD}_m^c} W^n(\by|X_m^n) \Bigg] \\
&\le\bbE_{X_m^n} \Bigg[ \frac{1}{M_n}\sum_{m=1}^{M_n}\sum_{\by} W^n(\by|X_m^n) \bone\{\by\in\tilde{\calD}_m^c\}\Bigg]  + \bbE_{\calC_n} \Bigg[ \frac{1}{M_n}\sum_{m=1}^{M_n}\sum_{\by}\sum_{m'\ne m} W^n(\by|X^n_m) \bone\{\by \in \tilde{\calD}_{m'} \}\Bigg] \label{eqn:split} \\
&\le\bbE_{X_m^n} \left[ \frac{1}{M_n}\sum_{m=1}^{M_n}\sum_{\by} W^n(\by|X_m^n) \bone\{\by\in\tilde{\calD}_m^c\}\right] + \exp(-bn^{1-t}), \label{eqn:intro_clt0}\\
&= \sum_{ \bx,\by}
P_X^n(\bx) W^n(\by|\bx) \bone
\left\{ 
 \log \frac{W^n(\by|\bx)}{P_Y^n (\by)} -nC 
< -(a-b) n^{1-t} \right\} + \exp(-bn^{1-t}), \label{eqn:intro_clt}
\end{align}
where \eqref{eqn:split} follows from \eqref{eqn:sets}, \eqref{eqn:intro_clt0} follows from similar calculations that led  to~\eqref{6-6-4}, and \eqref{eqn:intro_clt} follows from the definition of $\tilde{\calD}_m$ and the choice of $M_n$ in \eqref{eqn:sizeM}.
In fact, by using the bound $\hat{\calD}_m^c \supset\tilde{\calD}_m^c$ from the first equality in~\eqref{eqn:sets}, we see that the upper bound on  $\bbE_{\calC_n}  [ \Pr(\calE_1|\calC_n)  ]$  in \eqref{eqn:intro_clt} is tight in the sense that it can also be lower bounded as 
\begin{equation}
\bbE_{\calC_n} \big[ \Pr(\calE_1|\calC_n) \big]\ge\sum_{ \bx,\by}
P_X^n(\bx) W^n(\by|\bx) \bone
\left\{ 
 \log \frac{W^n(\by|\bx)}{P_Y^n (\by)} -nC 
<-(a-b)n^{1-t} \right\}. \label{eqn:lower_bd} 
\end{equation}
Recall that $a>b$. By  the moderate deviations theorem~\cite[Thm.~3.7.1]{Dembo}, 
the sums on the right-hand-sides of \eqref{eqn:intro_clt}  and \eqref{eqn:lower_bd} behave as
\begin{equation}
\exp\bigg(-  n^{1-2t} \, \frac{(a-b)^2}{2 U(P_X,W) }   + o(n^{1-2t}) \bigg),\label{6-6-5}
\end{equation}
which is much larger than  (i.e., dominates) the second term in \eqref{eqn:intro_clt}, namely  $\exp(- b n^{1-t})$.
Since $U(P_X,W) = V_{\min}(W)$  \cite[Lem.~62]{PPV10},  we have the asymptotic equality in \eqref{eqn:e1_md_direct}.

To derandomize the code, fix $\theta\in (0,1)$. By 
employing Markov's inequality  to \eqref{6-6-4} 
and \eqref{6-6-5} (cf.\ the proof of \cite[Thm.~1]{TanMoulin14}), we obtain
\begin{align}
\Pr \left(  \Pr(\calE_1|\calC_n)> \frac{1}{\theta}\bbE_{\calC_n} \big[ \Pr(\calE_1|\calC_n) \big] \right) & <\theta ,\quad\mbox{and} \label{eqn:derandom1}\\
\Pr \left(  \Pr(\calE_2|\calC_n)> \frac{1}{1-\theta}\bbE_{\calC_n} \big[ \Pr(\calE_2|\calC_n) \big] \right) & < 1-\theta.
\end{align}
Thus, 
\begin{equation}
\Pr \left(  \Pr(\calE_1|\calC_n)> \frac{1}{\theta}\bbE_{\calC_n} \big[ \Pr(\calE_1|\calC_n) \big]\mbox{ or }   \Pr(\calE_2|\calC_n)> \frac{1}{1-\theta}\bbE_{\calC_n} \big[ \Pr(\calE_2|\calC_n) \big]  \right)< 1.
\end{equation}
Thus,  by taking $\theta=\frac{1}{2},$ there exists a deterministic code satisfying
%
\begin{align}
\Pr(\calE_1|\calC_n)  &\le 2 \exp\bigg(-  n^{1-2t} \, \frac{(a-b)^2}{2 U(P_X,W) }   + o(n^{1-2t}) \bigg),\quad\mbox{and}\\
\Pr(\calE_2|\calC_n) & \le 2 \exp(- b n^{1-t}). \label{eqn:derandom-end}
\end{align}
 This completes the proof.

\subsection{Proof of Theorem~\ref{thm2}}\label{sec:prf2}
In this case, $t=1/2$. We first consider the case where $a \le 0$. Choose $P_X$ that achieves $V_{\max}(W)$. 
In this case by the Berry-Esseen theorem~\cite[Sec.~XVI.7]{feller}, the right-hand-sides of \eqref{eqn:intro_clt} and \eqref{eqn:lower_bd} behave as
\begin{equation}
\Phi\bigg( \frac{b-a}{\sqrt{V_{\max}(W)}}\bigg)+ O\bigg( \frac{1}{\sqrt{ n} }\bigg).
\end{equation}
Thus, by the same   Markov inequality argument to derandomize the code as above, for any sequence $\{\theta_n \}_{n\in\bbN} \subset (0,1)$, there exists a sequence of deterministic codes $\calC_n$ satisfying 
\begin{align}
\Pr(\calE_1|\calC_n)  &\le \frac{1}{\theta_n} \left[\Phi\bigg( \frac{b-a}{\sqrt{V_{\max}(W)}}\bigg) + O\bigg( \frac{1}{\sqrt{ n} }  \bigg)\right]\quad\mbox{and} \label{eqn:derandom_clt}\\
\Pr(\calE_2|\calC_n) &\le \frac{1}{1-\theta_n}\exp(-\sqrt{n}b).\label{eqn:derandom_clt2}
\end{align}
Choose $\theta_n :=1-1/n$ to complete the proof of the theorem for $a \le 0$. For $a >0$, choose the input distribution $P_X$ to achieve $V_{\min}(W)$ and proceed in exactly the same way.

\subsection{Proof of Lemma \ref{lem:md}}\label{sec:prf-md}
\begin{proof}
We consider choosing each codeword $\bx_m, m \in \{1,\ldots, {M_n}\}$ uniformly at random from $\{0,1,\ldots, d-1\}^n$. Indeed, a capacity-achieving input distribution of the additive channel is the uniform distribution on $\{0,1,\ldots, d-1\}$. As above, the expectation over this random choice of codebook is denoted as $\bbE_{\calC_n}[ \cdot]$. Now, we first consider $\Pr(\calE_1 |\calC_n)$. From the definition in~\eqref{eqn:p1}, the expectation of the error probability over the random codebook can be written as 
\begin{align}
\bbE_{\calC_n} \big[ \Pr(\calE_1|\calC_n) \big] & = 
\bbE_{\calC} \left[ \frac{1}{{M_n}}\sum_{m=1}^{M_n} \sum_{\by} W^n(\by|X^n_m) \bone\left\{ \frac{ \sum_{m'\ne m} W^n(\by|X^n_{m'}) }{W^n(\by|X^n_m)  }\ge 
\exp(-n{T_n})
\right\}\right]\\
& = \bbE_{\calC_n} \left[ \frac{1}{{M_n}}\sum_{m=1}^{M_n}  
\Pr\Bigg( \log   \bigg( \sum_{m'\ne m} W^n(Y^n|X^n_{m'})\bigg)  - \log W^n(Y^n|X^n_m) \ge 
-n{T_n}
 \,  \bigg|\, \calC_n \Bigg)\right]  \label{eqn:inner1}\\ 
& =  \frac{1}{{M_n}}\sum_{m=1}^{M_n}  \Pr\Bigg(\log  \bigg(  \sum_{m'\ne m} W^n(Y^n|X^n_{m'}) \bigg) - \log W^n(Y^n|X^n_m)\ge -b n^{1-t}
  \Bigg)\label{eqn:outer}     
\end{align}
In \eqref{eqn:inner1}, the inner probability is over $Y^n\sim W^n(\cdot|\bx_m)$ for a fixed code $\calC_n$ and in \eqref{eqn:outer}, the probability is over both the random codebook $\calC_n$ and the channel output $Y^n$ given message $m$ was sent. By symmetry of the codebook generation, it is sufficient to study the behavior of the random variable
\begin{equation}
F_n:= \log   \bigg(\sum_{m'\ne m} W^n(Y^n|X^n_{m'}) \bigg) 
- \log W^n(Y^n|X^n_m) \label{eqn:defZ}
\end{equation}
for any $m\in\{1,\ldots, {M_n}\}$, say $m=1$.   
In particular, to estimate the probability $\Pr(F_n\ge -bn^{1-t})$ in \eqref{eqn:outer}, it suffices to estimate the cumulant generating function of $F_n$.  We denote the cumulant generating function  as
\begin{align}
\phi_n(s) &:= \log\bbE \big[ \exp(sF_n) \big]  \label{eqn:phi1} \\
&= \log \bbE_{\calC_n} \bigg[ 
\sum_{\by} W^n(\by | X_m^n)^{1-s} \bigg(\sum_{m'\ne m}  W^n(\by | X_{m'}^n) \bigg)^s \bigg] \label{eqn:phi2} \\
&=  \log \sum_\by  \bbE_{\calC_n} \left[  W^n(\by | X_m^n)^{1-s}\right] \cdot    \bbE_{\calC_n} \bigg[\bigg(\sum_{m'\ne m}  W^n(\by | X_{m'}^n) \bigg)^s \bigg]  . \label{eqn:phi3} 
\end{align}
The final equality follows from the independence in the  codeword generation procedure. We have the following important lemma which is proved in Section~\ref{sec:prf6}.
\begin{lemma}[Asymptotics of Cumulant Generating Functions]\label{lem:6}
Fix $t\in (0,1/2]$.
Given the condition on the code size in \eqref{eqn:sizeM2}, 
the cumulant generating function satisfies 
\begin{equation}
\phi_n\bigg( \frac{u}{n^t}\bigg)=
\bigg(-au    + u^2   \frac{V(P)}{2} \bigg)n^{1-2t}
+O(n^{1-3t})  + o(1) 
\label{eqn:final_cgf}
\end{equation}
for  any constant $u > 0$.
\end{lemma}

Now, we apply the G\"artner-Ellis theorem with the general order, i.e., Case
(ii) of Theorem \ref{thm7} in Appendix~\ref{app:ge}, to \eqref{eqn:outer}    with 
the identifications 
\begin{align}
\alpha_n &\equiv 0\\
\beta_n &\equiv n^{1-t},\quad\mbox{and}\\
\gamma_n&\equiv n^{-t}
\end{align}
Now,  we can also make the additional identifications   
\begin{align}
\nu_1&\equiv 0 \\
X_n &\equiv -F_n\\
p_n(\cdot) & \equiv \frac{1}{{M_n}}\sum_{m=1}^{M_n}  \Pr(\cdot),\\
\mu_n(\theta )&\equiv\phi_n(-\theta),\\
\theta_0 &\equiv 0,\quad\mbox{and}\\
x &\equiv b.
\end{align}
Now, applying   \eqref{eqn:final_cgf} of Lemma \ref{lem:6} to the case with $u\equiv -y$ (with $y<0$), we have   
\begin{equation}
\nu_{2}(y)=\lim_{n\to\infty}\nu_{2,n}(y)=    y a+ y^2 \frac{ V(P)}{2}.
\end{equation}
 Then, defining $y_0$ to the (unique) real number satisfying $\nu_2'(y_0)=x$, we have 
\begin{equation}
a+ y_0 V(P)=b \quad\Longleftrightarrow\quad  y_0= \frac{b-a}{V(P)}  \label{eqn:y0}.
 \end{equation} 
 which implies by simple algebra that 
\begin{equation}
y_0 x-\nu_2(y_0)= \left( \frac{b-a}{V(P)} \right)b - \left[ \left(\frac{b-a}{V(P)} \right)a+  \left(\frac{b-a}{V(P)} \right)^2 \frac{V(P)}{2}\right]=\frac{(b-a)^2}{2V(P)}.
\end{equation}
Because $a>b$, $y_0$ is negative. We can also verify that all the conditions of Case (ii) in Theorem \ref{thm7} ($\theta_0=0,y_0<0, \alpha_n=\nu_1=0,\beta_n\gamma_n\to\infty$)  are satisfied so we can readily apply it here. 
Thus, 
\begin{align}
-\log\bbE_{\calC_n} \big[ \Pr(\calE_1|\calC_n) \big]  
= -\log\Pr\big( F_n > -b n^{1- t} \big)  
= \frac{(a-b)^2}{2V(P)} n^{1-2t}
+o(n^{1-2t}),
\label{eqn:exptC_1} 
\end{align}
which implies \eqref{eqn:e1_md}.


Now we estimate $\bbE_{\calC_n}[\Pr(\calE_2|\calC_n)]$.  Using the same calculations that led to \eqref{eqn:outer}, one finds that 
\begin{align}
\bbE_{\calC_n} \big[ \Pr(\calE_2|\calC_n) \big]&  
= \bbE_{\calC_n} \left[\frac{1}{{M_n}}\sum_{m=1}^{M_n} \sum_{\by} \sum_{m'\ne m} W^n(\by|X_{m'}^n)  \bone\left\{ \frac{ \sum_{m'\ne m} W^n(\by|X^n_{m' }) }{W^n(\by|X^n_m)  }
 <  \exp(-n{T_n})
\right\} \right] \label{eqn:pe2_0}\\
&=\bbE_{\calC_n} \left[  \frac{1}{{M_n}}\sum_{m=1}^{M_n}   \bbQ\bigg( \bigg\{ \by:  \log   \sum_{m'\ne m} W^n(\by|X^n_{m'})  - \log W^n(\by|X^n_m)  < 
-b n^{1-t}
 \bigg\}\,  \bigg|\, \calC_n  \bigg) \right] \label{eqn:pe2}  
\end{align}
where in \eqref{eqn:pe2}, we defined the (unnormalized) conditional measure $\bbQ (\calA|\calC_n = \{\bx_m\}_{m=1}^{M_n}) := \sum_{m'\ne m} W^n(\calA|\bx_{m'})$ where $\calA\subset\calY^n$.  
Given $\bbQ$, we can define a  {\em normalized probability} measure 
\begin{equation}
\bbQ'(\calA |\calC_n):= \frac{\bbQ(\calA|\calC_n)}{ M_n-1}.\label{eqn:norm_pm}
\end{equation}
Since the form of \eqref{eqn:pe2} is similar to the starting point for the calculation of $\bbE_{\calC_n}[\Pr(\calE_1|\calC_n)]$  in \eqref{eqn:outer}, we may estimate $\bbE_{\calC_n}[\Pr(\calE_2|\calC_n)]$ using similar steps to the above.  Define another probability measure $\bbP(\calA|\calC_n = \{\bx_m\}_{m=1}^{M_n}) :=W^n(\calA|\bx_m)$. Note by the definition of $F_n$ in \eqref{eqn:defZ}, and the measures above that for all $\calA\subset\calY^n$, 
\begin{equation}
\exp(F_n) = 
\frac{\bbQ'( \{Y^n\} |\calC_n)}{\bbP( \{Y^n\}|\calC_n)}\cdot (M_n-1). \label{eqn:com}
\end{equation}
Observe that the random variable involved in \eqref{eqn:pe2}, namely $\log   \sum_{m'\ne m} W^n(Y^n|X^n_{m'})  - \log W^n(Y^n|X^n_m)$, is exactly $F_n$ defined in \eqref{eqn:defZ} where $Y^n$ now has conditional law $\bbQ(\cdot|\calC_n)$  instead of $\bbP(\cdot|\calC_n)$. 
The cumulant generating function of $F_n$ under the probability measure $\bbQ'$ is 
\begin{align}
\lambda_n(s) & :=\log \bbE_{\calC_n,\bbQ'} \big[ \exp(sF_n)\big] \\
&=\log \bigg( \frac{\bbE_{\calC_n,\bbP} \big[\exp((1+s)F_n)\big]
}{M_n-1}\bigg)\label{eqn:use_com} \\
&= \phi_n(1+s )  - \log (M_n -1) \label{eqn:til_phi}
\end{align}
where \eqref{eqn:use_com} follows from \eqref{eqn:com} and \eqref{eqn:til_phi} from the definition of $\phi_n(s)$ in \eqref{eqn:phi1}. 
Now, we apply Case (ia)  of Theorem \ref{thm7} in Appendix \ref{app:ge} with the identifications 
\begin{align}
\alpha_n &\equiv -\log (M_n-1), \label{eqn:alpha_set}\\
\beta_n &\equiv n^{1-t} ,\quad\mbox{and}\\
\gamma_n &\equiv n^{-t}.
\end{align}
Furthermore, from \eqref{eqn:pe2}   and \eqref{eqn:til_phi}, one can also make the additional identifications   
\begin{align}
X_n  &\equiv F_n ,\\ 
p_n(\cdot)  &\equiv \bbE_{\calC_n}[\bbQ'(\cdot|\calC_n)],\\
 \mu_n(\theta)&\equiv\lambda_n(\theta), \\
 \theta_0&\equiv -1,\\
 \nu_1& \equiv 0,  \quad\mbox{and}\\
  x&\equiv -b.
\end{align}
Then we have 
\begin{equation}
\nu_{2,n}(y)\equiv \lambda_n(-1+ y n^{-t})+ \log (M_n-1)=  n^{2t-1}\phi_n \left( \frac{y}{n^{ t}}\right)
\end{equation} 
Thus, using \eqref{eqn:final_cgf},
\begin{equation}
\nu_{2}(y)=\lim_{n\to\infty}\nu_{2,n}(y)=  - y a+ y^2 \frac{V(P)}{2} . \label{eqn:limit_nu2}
\end{equation}   So, we have $\theta_0  x-\nu_1 =b$
and 
\begin{equation}
y_0= \frac{a-b}{V(P)},
\end{equation}
  which is positive. This implies that
\begin{equation}
y_0 x -\nu_2(y_0)=\frac{(a-b)^2}{2V(P)}.
\end{equation}
 Thus,  by the relation between $\bbQ$ to $\bbQ'$ in \eqref{eqn:norm_pm} and the bound in   \eqref{ge1_upper} (in Case (ia) of Theorem \ref{thm7}), we obtain 
\begin{align}
 -\log\bbE_{\calC_n} \big[ \Pr(\calE_2| \calC_n) \big]  
&=-\log\bbE_{\calC_n}\left[\bbQ\bigg( F_n <  -b n^{1-t}\,\bigg|\, \calC_n\bigg)\right] \label{eqn:exptC_20} \\
&= 
-\log\bbE_{\calC_n}\left[\bbQ'\bigg( F_n <  -b n^{1-t}\,\bigg|\, \calC_n\bigg)\right] -\log (M_n-1) \\
&\le b n^{1-t}+
 o(n^{1- t}),
\label{eqn:exptC_2}
\end{align}
which implies the upper bound in~\eqref{eqn:e2_md}. The lower bound in \eqref{eqn:e2_md} follows by invoking \eqref{ge1_lower}, from which we obtain 
\begin{equation}
 -\log\bbE_{\calC_n} \big[ \Pr(\calE_2| \calC_n) \big]  \ge b n^{1-t}+ \frac{(a-b)^2}{2V(P)} n^{1-2t} +
 o(n^{1- 2t}). \label{eqn:lower_higher}
\end{equation}
This completes the proof  of Lemma~\ref{lem:md}.
\end{proof}

\begin{remark}
Observe that to evaluate the probabilities in  \eqref{eqn:exptC_1} and \eqref{eqn:exptC_20}, we employed Theorem~\ref{thm7} in Appendix \ref{app:ge}, which is a modified (``shifted'') version of the usual G\"artner-Ellis theorem~\cite[Theorem 2.3.6]{Dembo}. Theorem \ref{thm7} assumes a sequence of random variables $X_n$ has cumulant  generating functions $\mu_n(\theta)$ that additionally satisfy the expansion $\mu_n(\theta_0 + \gamma_n y) = \alpha_n + \beta_n\nu_1 + \beta_n\gamma_n\nu_{2,n}(y)$ for some vanishing sequence  $\gamma_n$.   This generalization and the application to the erasure problem  appears to the authors to be  novel. In particular, since $\bbQ$  in \eqref{eqn:pe2} above is not a (normalized) probability measure, the usual G\"artner-Ellis theorem does not apply readily and we have to define the new probability measure $\bbQ'$ as in \eqref{eqn:norm_pm}. This, however, is not the crux of the contributions of which there are three. 
\begin{enumerate}
\item First, our Theorem \ref{thm7} also has to take into account the   offsets $\theta_0=-1$ and $\alpha_n=-\log (M_n-1)$ in our application of the G\"artner-Ellis theorem. 
\item Second,    an interesting feature of our result is that   the ``exponent'' $b$ is not governed by the first-order term $\alpha_n$ (which is the offset) but instead the second-order term $-(\theta_0  x-\nu_1)\beta_n=-bn^{1-t}$  leading to \eqref{eqn:exptC_2}--\eqref{eqn:lower_higher}.  
  \item Finally, Theorem~\ref{thm7} also allows us to obtain an additional term scaling as   $n^{1-2t}$ in \eqref{eqn:lower_higher}, but we cannot obtain the coefficient of the higher-order term scaling as $n^{1-2t}$ in \eqref{eqn:exptC_2}.
\end{enumerate} 
\end{remark}
\subsection{Proof of Theorem \ref{thm:md}} \label{sec:optimality_prf}

\begin{proof} 
First, observe that for~\eqref{H3} to be satisfied, i.e., that the undetected error probability decays as 
\begin{equation}
\bbE_{\calC_n}\big[\Pr(\calE_2 | \calC_n) \big] \le\exp(- b n^{1-t} + o(n^{1-t})), \label{eqn:E2_speed}
\end{equation}
we need the threshold to be of the form 
\begin{equation}
T_n\ge b n^{-t} + o(n^{-t} ). \label{eqn:Tn_greater}
\end{equation}
 To show this formally, suppose, to the contrary, 
\begin{equation}
 T_n = b'n^{-t} + o(n^{-t} ).
 \end{equation} 
 for some $0<b'<b$. So the constant in front of $n^{ -t}$ is strictly smaller than $b$. Then by \eqref{eqn:e2_md} in Lemma \ref{lem:md}, and since the decoder is asymptotically optimal, the undetected error probability decays as 
 \begin{equation}
\bbE_{\calC_n}\big[\Pr(\calE_2 | \calC_n) \big] = \exp(- b' n^{1-t} + o(n^{1-t})),
 \end{equation}
 an asymptotic equality. This is a contradiction. Hence,  \eqref{eqn:Tn_greater} must hold. Intuitively, the thresholds in Forney's test in \eqref{eqn:decision_reg} must be large enough so that the decoding regions corresponding to the messages are sufficiently small so that the undetected error probability decays at least as fast as  in \eqref{eqn:E2_speed}.   Consequently, by the asymptotically tight result in~\eqref{eqn:e1_md}, we know that \eqref{H3} implies \eqref{H4}.

Conversely, to satisfy the condition  \eqref{H1}, i.e., that the total error  probability decays as 
\begin{equation}
\bbE_{\calC_n}\big[\Pr(\calE_1 | \calC_n) \big] \le\exp \left(- \frac{(a-b)^2}{2V(P)} n^{1-2t} + o(n^{1-2t}) \right),
\end{equation}
 we need to choose 
\begin{equation}
 T_n\le {b}n^{-t} + o(n^{-t}) 
 \end{equation} 
by the same  argument as the above. This means that  the thresholds must be small enough so that the decoding regions corresponding to the messages are sufficiently large. Hence, due to the asymptotically tight result in Lemma~\ref{lem:md}, we have \eqref{H2}. 
\end{proof}

\subsection{Proof of Lemma~\ref{lem:clt}}\label{sec:prf-clt}
\begin{proof}
The exact same steps in the proof of Lemma~\ref{lem:md}  follow  even if $t=1/2$. In particular,  in this setting, Lemma \ref{lem:6} with $t=1/2$ yields
\begin{equation}
\lim_{n\to\infty}  \,\, \phi_n\bigg(\frac{u}{\sqrt{n}}\bigg)
= -u a    + u^2   \frac{V(P)}{2}  \label{eqn:det_cumulants}
\end{equation}
for any constant $u > 0$.   By appropriate translation, scaling,  and Lemma \ref{l625} in Appendix~\ref{app:conv},   the sequence of random variables $\{F_n n^{-1/2}\}_{n\in\bbN}$  converges in distribution to a Gaussian random variable with mean $-a$ and variance $V(P)$.  This implies  that the following asymptotic statement holds true 
\begin{align}
\lim_{n\to\infty}\,\, \bbE_{\calC_n}[\Pr(\calE_1 | \calC_n)] 
&=\lim_{n\to\infty}\,\, \Pr\bigg(\frac{F_n}{\sqrt{n}}>-b \bigg)\\
&=\int_{-b}^{\infty}\frac{1}{\sqrt{2\pi V(P)} }\exp \bigg(-\frac{ (w  +a )^2}{2 \, V(P ) } \bigg)\,\rmd w \\
&=  \Phi\bigg( \frac{b-a}{\sqrt{V(P)}}\bigg).
\end{align}
To calculate   $\bbE_{\calC_n}[\Pr(\calE_2 | \calC_n)]$, we can adopt the same change of measure and G\"artner-Ellis arguments and employ Case (ib) of  Theorem \ref{thm7}. We follow exactly  the steps leading from \eqref{eqn:pe2_0} to~\eqref{eqn:exptC_2} to assert that  \eqref{eqn:e2_clt}  is true. Note that in this situation, we take $\gamma_n \equiv n^{-1/2}$ and $\beta_n \equiv n^{1/2}$. To apply Case (ib) of Theorem \ref{thm7}, we verify that $\beta_n=\gamma_n^{-1}$ and $\nu_2$, derived in \eqref{eqn:limit_nu2}, is indeed a quadratic function. 
\end{proof}

\subsection{Proof of Lemma~\ref{lem:6}: Asymptotics of Cumulant Generating Functions}\label{sec:prf6}

\begin{proof}
To estimate  $\phi_n(s)$  in \eqref{eqn:phi3}, we define
\begin{align}
A &:=\bbE_{\calC_n} \left[  W^n(\by | X_m^n)^{1-s}\right] \quad\mbox{and} \label{eqn:defA}\\
B&:= \bbE_{\calC_n} \bigg[\bigg(\sum_{m'\ne m}  W^n(\by | X_{m'}^n) \bigg)^s \bigg]. \label{eqn:defB}
\end{align} 

The first term $A$ is easy to handle. Indeed, by the additivity of the channel, we have
\begin{align}
A& = \bbE_{X_m^n} \left[ P^n (\by- X_m^n)^{1-s} \right]  \\
&= \bbE_{X_m^n} \big[  P^n (\tilX_m^n)^{1-s} \big] 
\end{align}
where the {\em shifted codewords}\footnote{The shifted codewords need not be codewords  {\em per se}, so this is a slight abuse of terminology.} are defined as $\tilX^n_m :=\by - X^n_m$. By using the product structure of $P^n$, we see that regardless of $\by$, the term $A$  can be written as 
\begin{equation}
A = \frac{1}{d^n}\exp\big( -n \psi(s)\big)
\end{equation}
where 
\begin{equation}
\psi(s) := -\log \sum_{z} P(z)^{1-s}.
\end{equation}
This function is related to the R\'enyi entropy as follows: $s\psi(s)  = -H_{1-s}(P)$ where $H_{\alpha}(P)$ is the usual R\'enyi entropy of order $\alpha$ (e.g., \cite[Prob.~1.15]{Csi97}). Now, for a fixed $u  >  0$, we make the choice 
\begin{equation}
s=\frac{u}{n^t}, \label{eqn:defs}
\end{equation} 
where recall that $t$ is a fixed parameter in $(0,1/2]$. It is straightforward to check that $\psi(0)=0$, $\psi'(0) = -H(P)$ and $\psi''(0)=-V(P)$.  By  a second-order Taylor expansion of $\psi(s)$ around $s=0$, we have 
\begin{align}
A&=  \frac{1}{d^n} \exp\bigg( n   \Big( s\, H(P) + s^2\, \frac{V(P)}{2}   + O(s^3) \Big)\bigg) \\
&=\frac{1}{d^n} \exp\bigg( u n^{1-t}\,  H(P) +u^2 n^{1-2t} \,\frac{V(P)}{2}  + O( n^{1-3t}) \bigg) , \label{eqn:est_A}
\end{align}
where \eqref{eqn:est_A} follows from the definition of $s$ in \eqref{eqn:defs}.

Now we estimate $B$ in \eqref{eqn:defB}. Define  the random variable $N_{\calC_n}  (Q)$ which represents the number of shifted codewords excluding that indexed by $m$  with type $Q\in\calP_n(\calX)$, i.e., $N_{\calC_n}  (Q):=|\{ m'\ne m: \mathrm{type}(\tilX^n_{m'})=Q\}|$. This plays the role of the {\em type class enumerator} or {\em distance enumerator} in Merhav \cite{merhav08, merhav_FnT}.  Then, $B$  can be written as 
\begin{align}
B  &= \bbE_{\calC_n} \bigg[\bigg(\sum_{m'\ne m}   P^n(\by - X_{m'}^n) \bigg)^s \bigg] \label{eqn:use_add0}\\
 &= \bbE_{\calC_n} \bigg[\bigg(\sum_{m'\ne m}    P^n(\tilX_{m'}^n) \bigg)^s \bigg]  \label{eqn:use_add}\\
  &= \bbE_{\calC_n} \bigg[\bigg(\sum_{Q\in\calP_n(\calX)}   N_{\calC_n}  (Q) \exp\big(-n[D(Q\|P) + H(Q) ] \big) \bigg)^s \bigg] . \label{eqn:type_class_enum}
\end{align}
In \eqref{eqn:use_add0}, we again used the additivity of the channel and introduced the noise distribution $P$. In \eqref{eqn:use_add}, we   used the definition of the shifted codewords $\tilX^n_m$. In \eqref{eqn:type_class_enum}, we introduced the type class enumerators $N_{\calC_n}  (Q)$. We also recall from  \cite[Lem.~2.6]{Csi97}  that $\exp (-n[D(Q\|P) + H(Q) ] ) $ is the exact $P^n$-probability of a sequence of type $Q$. Note that the expression in \eqref{eqn:type_class_enum} is independent of $\by$, just as for the  calculation of $A$ in \eqref{eqn:est_A}. In the following, we find   bounds on $B$ that turn out to tight in the sense that the analysis yield the final result in Theorem \ref{thm:md}. We start with lower bounding $B$ by as follows:
\begin{align}
B &\ge  \bbE_{\calC_n} \bigg[\bigg(\max_{Q'\in\calP_n(\calX)}   N_{\calC_n}  (Q') \exp\big(-n[D(Q'\|P) + H(Q') ] \big) \bigg)^s \bigg] \\
&=  \bbE_{\calC_n} \bigg[ \max_{Q'\in\calP_n(\calX)}   N_{\calC_n}  (Q')^s \exp\big(-ns[D(Q'\|P) + H(Q') ] \big)  \bigg] \\
&\ge \max_{Q'\in\calP_n(\calX)} \bbE_{\calC_n} \big[ N_{\calC_n}  (Q')^s \big] \exp\big(-ns[D(Q'\|P) + H(Q')] \big)\\
&\ge \bbE_{\calC_n} \big[ N_{\calC_n}  (P_n)^s \big] \exp\big(-ns[D(P_n\|P) + H(P_n)] \big), \label{eqn:max-type}  
\end{align}
where $P_n\in \calP_n(\calX)$ is defined as
\begin{equation}
P_n \in\argmin_{Q \in \calP_n(\calX)} \{
\|Q -P\|_1: H(Q) \ge H(P)+ 2 a n^{-t} \}.
\label{eqn:choice_delta}
\end{equation}
Fannes inequality \cite[Lem.~2.7]{Csi97} (uniform continuity of Shannon entropy) says that 
\begin{equation}
|H(P) - H(Q)|\le \|P-Q\|_1\log \frac{\|P-Q\|_1}{|\calX|}
\end{equation}
if $\| P-Q\|_1\le\frac{1}{2}$. Since $P_n$ must be an $n$-type,  
\begin{equation}
H(P_n) = H(P)+ 2 a n^{-t} +O( n^{-1}\log n). \label{eqn:apply_fannes}
\end{equation}
Because $P(z) >0$ for all $z\in\calX$, one immediately finds that 
\begin{equation}
D(P_n\| P) = O(\|P_n-P\|_1^2)= O(n^{-2t} ),
\end{equation}
which is negligible. Combining the above estimates, we obtain
\begin{align}
-ns[D(P_n\|P) + H(P_n)] 
=- u n^{1-t} H(P) -2 a u n^{1-2t} + O(n^{1-3t} )
\label{eq6-21-2}
\end{align}
as $n$ grows. 

Next, apply Lemma \ref{l6-21} in Appendix \ref{app:conc}  to the expectation  \eqref{eqn:max-type}  to the case with
\begin{align}
L &=M_n-1, \label{eqn:L}\\
M_1 &=d^n, \label{eqn:M1}\\
M_2 &= |{\cal T}^{(n)}_{P_n}|,\\
\{X_1, \ldots, X_L\} &=\{X^n_{m'}\}_{m'\neq m},\\
\calA &= {\cal T}^{(n)}_{P_n},\quad\mbox{and},\\
s &={u}  {n^{-t}} \label{eqn:define_s}
\end{align} 
and a fixed positive constant $\epsilon>0$. We now perform a series of steps to bound the terms in  \eqref{eqn:conc_ineq}.
By a standard property of types~\cite{Csi97} and the estimate  in \eqref{eqn:apply_fannes},
\begin{align}
\log |{\cal T}^{(n)}_{P_n}| & \ge nH(P_n)-(d-1)\log (n+1) \\
&=  n H(P) 
+2 a  n^{1- t}
+O(\log n).
\end{align}
Thus, using the definition of $L$ in \eqref{eqn:L}, the definition $M_1$ in \eqref{eqn:M1}, and the number of codewords $M_n$ in \eqref{eqn:sizeM2}, we also have
\begin{equation}
\log L +\log M_2 - \log M_1 \ge a n^{1-t} +O(\log n ). \label{eqn:logL}
\end{equation}
Consequently, 
\begin{equation}
\log  \left[1-\exp \left(-   \frac{LM_2 }{2M_1}  \epsilon^2 \right)  \right]=o(1).
\end{equation}
Also, we have
\begin{equation}
s \log (1-\epsilon )= {u}{n^{-t}}\log (1-\epsilon )=o(1).
\end{equation}
and by \eqref{eqn:define_s},
\begin{equation}
s(\log L +\log M_2 - \log M_1) \ge a u n^{1-2t} +o(1).
\end{equation}
Therefore, Lemma~\ref{l6-21} says that
\begin{align}
\log \bbE_{\calC_n} \big[ N_{\calC_n}  (P_n)^s \big]
= \log \bbE [N^s ] \ge  a u n^{1-2t} +o(1).
\label{eq6-21-3}
\end{align}
Combining \eqref{eqn:max-type}, \eqref{eq6-21-2} and \eqref{eq6-21-3},
we find that
\begin{align}
\log B \ge 
-n^{1-t} u H(P) - a n^{1-2t} u  + O(n^{1-3t} ) +o(1).
\label{eqn:lowerB}
\end{align}

Now, we proceed to upper bound $B$ in \eqref{eqn:type_class_enum}.  Note that we consider the case when $0 < s < 1$ because we substitute $un^{-t}$ into $s$ per \eqref{eqn:define_s}. Consider,
\begin{align}
B  
&=\bbE_{\calC_n}\bigg[  \bigg(
\sum_{Q' \in\calP_n(\calX)}  N_{\calC_n}  (Q') \exp\big(-n[D(Q'\|P) + H(Q') ] \big) \bigg)^s \bigg] \label{eqn:jsens0}   \\
&\le   \bigg(\bbE_{\calC_n}\bigg[
\sum_{Q' \in\calP_n(\calX)}  N_{\calC_n}  (Q') \exp\big(-n[D(Q'\|P) + H(Q') ] \big) \bigg] \bigg)^s \label{eqn:jsens}   \\
&=   \bigg(
\sum_{Q' \in\calP_n(\calX)}  
\bbE_{\calC_n}\big[N_{\calC_n} (Q')  \big]
\exp\big(-n[D(Q'\|P) + H(Q') ] \big)  \bigg)^s   \\
&=   \bigg(
\sum_{Q' \in\calP_n(\calX)}  
\frac{M_n |{\cal T}_{Q'}^{(n)}|}{d^n}
 \exp\big(-n[D(Q'\|P) + H(Q') ] \big)  \bigg)^s    \label{eqn:expct_N} \\
&\le   \bigg(
(n+1)^{d-1}
\max_{Q' \in\calP_n(\calX)}  
\frac{M_n |{\cal T}_{Q'}^{(n)}|}{d^n}
 \exp\big(-n[D(Q'\|P) + H(Q') ] \big)  \bigg)^s   \\
&=  
(n+1)^{s(d-1)}
\max_{Q' \in\calP_n(\calX)}  
\bigg(\frac{M_n |{\cal T}_{Q'}^{(n)}|}{d^n}\bigg)^s
\exp\big(-ns[D(Q'\|P) + H(Q') ] \big) 
\\
&\le   
(n+1)^{s(d-1)}
\max_{Q' \in\calP_n(\calX)}  
\exp\big(
-s[nH(P)-an^{1-t} -nH(Q')]
-ns[D(Q'\|P) + H(Q') ] \big)     \label{eqn:expct_N2}
\\
&=   
(n+1)^{s(d-1)}
\max_{Q' \in\calP_n(\calX)}  
\exp\big(
-s nH(P)-asn^{1-t} 
-ns D(Q'\|P)  \big)  
\\
&=
(n+1)^{s(d-1)}
\exp\bigg(
-s nH(P)-asn^{1-t}
+ \max_{Q' \in\calP (\calX)} -ns D(Q'\|P)  \bigg)  
\\
& \le
(n+1)^{s(d-1)}
\exp\big(
-s nH(P)-asn^{1-t}  \big)  \\
&=
(n+1)^{\frac{u(d-1)}{n^t}}
\exp\big(
-u n^{1-t} H(P)-a u n^{1-2t}  \big)  
\end{align}
where \eqref{eqn:jsens} follows from Jensen's inequality applied to the concave function  $x \mapsto x^s$, \eqref{eqn:expct_N} follows from the fact that 
\begin{equation}
\bbE_{\calC_n}\big[N_{\calC_n} (Q')\big] =M_n  \Pr\left(\by-X_1^n\in\calT_{Q'}^{(n)}\right) = M_n \Pr \left(\tilX_1^n\in\calT_{Q'}^{(n)} \right) =  \frac{M_n|\calT_{Q'}^{(n)}|}{d^n} \label{eqn:calc_N}
\end{equation}
and \eqref{eqn:expct_N2} follows from the choice of $\log M_n= n (\log d - H(P)) - a n^{1-t}$ and the fact that $|\calT_{Q'}^{(n)}|\le \exp(nH(Q'))$. 
Thus,
we find that
\begin{align}
\log B \le 
-n^{1-t} u H(P) - a n^{1-2t} u  +o(1).
\label{eqn:optimizing}
\end{align}
Combining the evaluations of $A$ and $B$ together in \eqref{eqn:phi3}, we see that  the sum over $\by$ cancels the $1/d^n$ term in~\eqref{eqn:est_A} and the first-order entropy terms also cancel. 
The final expression for the cumulant generating function of $F_n$ satisfies \eqref{eqn:final_cgf} as desired.
\end{proof}

\section{Conclusion and Future Work}\label{sec:concl}
In this paper, we analyzed channel coding with the erasure option where we designed both the undetected and total errors to decay subexponentially and asymmetrically. We analyzed two regimes, namely, the pure moderate deviations and mixed regimes. We proposed an information spectrum-type decoding rule \cite{Han10} and showed using an ensemble tightness argument that this simple decoding rule is, in fact, asymptotically optimal for additive DMCs with uniform input distribution. To do so, we estimated appropriate cumulant generating functions of the total and undetected errors. We also  developed a modified version of the G\"artner-Ellis theorem that is particularly useful for our problem. In contrast to previous works on erasure (and list) decoding \cite{Forney68,sgb, tel94, Blinovsky, Moulin09, merhav08,merhav13,merhav14,  somekh11}, we do not evaluate the rate of exponential decay of the two error probabilities. In our work, the two error probabilities decay subexponentially (and asymmetrically) for the pure moderate deviations setting. For the mixed regime, the total (and hence erasure) error is non-vanishing while the undetected error decays as $\exp(-bn^{1/2})$ for some $b>0$.  

Possible extensions of this work include:
\begin{enumerate}
\item  Removing the assumption that the DMC is additive for the ensemble tightness results in Section~\ref{sec:ens_results}. However, it appears  that this is not straightforward and it is likely that we have to make an assumption like that for Theorem 1 of Merhav's work~\cite{merhav08}. This assumption seems necessary using our techniques to establish the asymptotics of the cumulant generating function in Lemma \ref{lem:6}. It heavily relies on the fact that input distribution is uniform so that, by symmetry, the statistics of the codewords $\{X_m^n : m = 1,\ldots, M_n\}$ are the same as that of the shifted codewords $\{\by-X_m^n : m = 1,\ldots, M_n\}$ for any $\by$. 
\item Extending the analysis to the list decoding case where $T_n= {b} n^{-t}$ where $b<0$ and $t\in (0,1/2]$. Our  information spectrum-style decoding regions and subsequent analysis of their probabilities only works for the case $b>0$. See the argument after \eqref{6-6-5}. Hence it would be useful to develop alternative threshold decoders and more refined tools to analyze the list decoding setting.
\item Tightening the  higher-order asymptotics for the  expansions of the log-probabilities in \eqref{eqn:e2_md} and \eqref{eqn:e2_clt}. This would  be interesting from a mathematical standpoint. However, this appears to require some independence assumptions which are not available  in the G\"artner-Ellis theorem (so a new concentration bound may be required). In addition, this seems  to require tedious calculus to evaluate  the higher-order asymptotic terms of the cumulant generating function in Lemma~\ref{lem:6}. A refinement of the type class enumerator method~\cite{merhav08, merhav_FnT, merhav13,merhav14, somekh11} seems to be necessary for this purpose. 

\end{enumerate}
\appendices

\input{ge_prf}

\section{Convergence in Distribution based on Convergence of  Cumulant Generating Functions} \label{app:conv}
\begin{lemma}\label{l625}
Let $\mu_n$ be a sequence of probability measures on $\mathbb{R}$.
Suppose that  for some $0<a<b$,
\begin{equation}
 \log\int_{\mathbb{R}}
\exp(sx)\, \mu_n(\rmd x) \to f(s):=\frac{s^2}{2},\qquad \forall\,  s\in (a,b).\label{eqn:sgre}
\end{equation}
Then,  $\mu_n$ converges (weakly) to the standard Gaussian  $\mu(\calA) =\int_\calA\varphi(w)\,\rmd w$.
\end{lemma}
Notice that in \eqref{eqn:sgre}, the assumption pertains only to $s$ in the open interval $(a,b)$. In particular, it is not assumed that the convergence holds for all  $s\in\bbR$, in which case convergence of $p_n$ to the standard normal is an elementary fact (cf.~L\'evi's continuity theorem \cite[Thm.~18.21]{Fristedt}).  

See Mukherjea {\em et al.} \cite[Thm.~2]{muk06} for the proof of  Lemma \ref{l625}. 

\section{A Basic Concentration Bound} \label{app:conc}
\begin{lemma}\label{l6-21}
Let $X_1, \ldots, X_L$ be independent random variables, each distributed according to the uniform distribution
on $\{1, \ldots, M_1\}$.
We fix a subset $\calA \subset \{1, \ldots, M_1\}$ whose cardinality is $M_2$.
We denote the random number $|\{ i  \in \{1,\ldots, L\}: X_i \in \calA \}|$ by $N$.
For every $s>0$,
\begin{align}
\bbE [N^s ] \ge\left\lfloor \frac{LM_2}{M_1}(1-\epsilon)\right\rfloor^s
\left[ 1- \exp\bigg(- L\,   \frac{M_2}{2M_1}\epsilon^2\bigg) \right] \label{eqn:conc_ineq}
\end{align}
where    $0<\epsilon <1$ is also an arbitrary number. 

\end{lemma}

\begin{proof}
By straightforward calculations, we have
\begin{align}
\bbE [N^s ] &= \sum_{l=0}^{L} l^s \Pr(N = l ) \\ 
&\ge \sum_{l \ge LM_2(1-\epsilon)/M_1 } l^s \Pr(N = l ) \\ 
&\ge\left\lfloor\frac{LM_2}{M_1}(1-\epsilon) \right\rfloor^s\Pr \left(N \ge  \frac{LM_2(1-\epsilon)}{M_1}  \right)\\
& =\left\lfloor\frac{LM_2}{M_1}(1-\epsilon) \right\rfloor^s\left[1-\Pr \left( \frac{N}{L} <  (1-\epsilon) \frac{ M_2}{M_1}  \right)\right]  . \label{eqn:poisson_trials}
\end{align}
Now since the event in  probability in \eqref{eqn:poisson_trials} implies that  the relative frequency of the number of events $\{X_i \in \calA\},i = 1,\ldots, L$ is less than $1-\epsilon$ multiplied by the mean $\bbE[ \bone\{ X_i \in\calA\} ]=M_2/M_1$ of each indicator $\bone\{ X_i \in\calA\} $, we can invoke the Chernoff bound   for  independent Bernoulli trials (e.g.,~\cite[Thm.~4.5]{mitz}) to conclude that 
\begin{equation}
\Pr \left( \frac{N}{L} <  (1-\epsilon)\frac{ M_2}{M_1}  \right) = \Pr\left( \frac{1}{L} \sum_{i=1}^L \bone\{ X_i\in\calA\} < (1-\epsilon)\frac{ M_2}{M_1}  \right)\le\exp\left( -L\frac{M_2}{2M_1}\epsilon^2\right).\label{eqn:trials}
\end{equation}
Combining this with \eqref{eqn:poisson_trials} concludes the proof. 
\end{proof}

\subsubsection*{Acknowledgements}   
The authors would like to acknowledge the Associate Editor
(Prof.\ Aaron Wagner) and the anonymous reviewers for their
extensive and useful comments during the revision process.

MH is grateful for Prof.\ Nobuo Yoshida  for clarifying Lemma \ref{l625}.  VT is grateful to Prof.\ Rongfeng Sun   for clarifying the same lemma. 

\bibliographystyle{unsrt}
\bibliography{isitbib}

\begin{IEEEbiographynophoto}{Masahito Hayashi}(M'06--SM'13) was born in Japan in 1971.
He received the B.S. degree from the Faculty of Sciences in Kyoto
University, Japan, in 1994
and the M.S. and Ph.D. degrees in Mathematics from
Kyoto University, Japan, in 1996 and 1999, respectively.

He worked in Kyoto University as a Research Fellow of the Japan Society of the
Promotion of Science (JSPS) from 1998 to 2000,
and worked in the
Laboratory for Mathematical Neuroscience,
Brain Science Institute, RIKEN from 2000 to 2003,
and worked in ERATO Quantum Computation and Information Project,
Japan Science and Technology Agency (JST) as the Research Head from 2000 to 2006.
He also worked in the Superrobust Computation Project Information Science and Technology Strategic Core (21st Century COE by MEXT) Graduate School of Information Science and Technology, The University of Tokyo as Adjunct Associate Professor from 2004 to 2007.
In 2006, he published the book ``Quantum Information: An Introduction'' from Springer.
He worked in the Graduate School of Information Sciences, Tohoku University as Associate Professor from 2007 to 2012.
In 2012, he joined the Graduate School of Mathematics, Nagoya University as Professor.
He also worked in Centre for Quantum Technologies, National University of Singapore as Visiting Research Associate Professor from 2009 to 2012
and as Visiting Research Professor from 2012 to now.
In 2011, he received Information Theory Society Paper Award (2011) for ``Information-Spectrum Approach to Second-Order Coding Rate in Channel Coding''.

He is on the Editorial Board of {\it International Journal of Quantum Information}
and {\it International Journal On Advances in Security}.
His research interests include classical and quantum information theory and classical and quantum statistical inference.
\end{IEEEbiographynophoto}

\begin{IEEEbiographynophoto}{Vincent Y. F. Tan}
(S'07-M'11-SM'15)  was born in Singapore in 1981. He is currently an Assistant Professor in
the Department of Electrical and Computer Engineering (ECE) and the
Department of Mathematics at the National University of Singapore (NUS).
He received the B.A.\ and M.Eng.\ degrees in Electrical and Information
Sciences from Cambridge University in 2005 and the Ph.D.\ degree in
Electrical Engineering and Computer Science (EECS) from the Massachusetts
Institute of Technology in 2011. He was a postdoctoral researcher in the
Department of ECE at the University of Wisconsin-Madison and a research
scientist at the Institute for Infocomm (I$^2$R) Research, A*STAR, Singapore.
His research interests include information theory, machine learning and signal
processing.

Dr.\ Tan received the MIT EECS Jin-Au Kong outstanding doctoral thesis
prize in 2011 and the NUS Young Investigator Award in 2014.  He has authored a research monograph on {\em ``Asymptotic Estimates in Information Theory with Non-Vanishing Error Probabilities''} in the Foundations and Trends  in Communications and Information Theory Series (NOW Publishers).  He is currently 
an Associate Editor of the IEEE Transactions  on Communications. \end{IEEEbiographynophoto}

\end{document}

%% file: preamble.tex
\usepackage[mathscr]{eucal}
\usepackage[cmex10]{amsmath}
\usepackage{epsfig,epsf,psfrag}
\usepackage{amssymb,amsmath,amsthm,amsfonts,latexsym}
\usepackage{amsmath,graphicx,bm,xcolor,url}
\usepackage[caption=false]{subfig} 
\usepackage{fixltx2e}
\usepackage{array}
\usepackage{verbatim}
\usepackage{bm}
\usepackage{verbatim}
\usepackage{textcomp}
\usepackage{mathrsfs}
\usepackage{epstopdf}

\catcode`~=11 \def\UrlSpecials{\do\~{\kern -.15em\lower .7ex\hbox{~}\kern .04em}} \catcode`~=13 

\allowdisplaybreaks[1]
 
\newcommand{\nn}{\nonumber}

\newcommand{\calA}{\mathcal{A}}

\newcommand{\calC}{\mathcal{C}}
\newcommand{\calD}{\mathcal{D}}
\newcommand{\calE}{\mathcal{E}}

\newcommand{\calP}{\mathcal{P}}

\newcommand{\calT}{\mathcal{T}}

\newcommand{\calX}{\mathcal{X}}
\newcommand{\calY}{\mathcal{Y}}

\newcommand{\bx}{\mathbf{x}}

\newcommand{\by}{\mathbf{y}}

\newcommand{\rmd}{\mathrm{d}}

\newcommand{\rme}{\mathrm{e}}

\newcommand{\bbE}{\mathbb{E}}

\newcommand{\bbN}{\mathbb{N}}

\newcommand{\bbP}{\mathbb{P}}
\newcommand{\bbQ}{\mathbb{Q}}
\newcommand{\bbR}{\mathbb{R}}

\DeclareMathAlphabet{\mathbsf}{OT1}{cmss}{bx}{n}
\DeclareMathAlphabet{\mathssf}{OT1}{cmss}{m}{sl}

\newcommand{\rvd}{\mathsf{d}}

\newcommand{\rvf}{\mathsf{f}}

\DeclareSymbolFont{bsfletters}{OT1}{cmss}{bx}{n}  
\DeclareSymbolFont{ssfletters}{OT1}{cmss}{m}{n}
\DeclareMathSymbol{\bsfGamma}{0}{bsfletters}{'000}
\DeclareMathSymbol{\ssfGamma}{0}{ssfletters}{'000}
\DeclareMathSymbol{\bsfDelta}{0}{bsfletters}{'001}
\DeclareMathSymbol{\ssfDelta}{0}{ssfletters}{'001}
\DeclareMathSymbol{\bsfTheta}{0}{bsfletters}{'002}
\DeclareMathSymbol{\ssfTheta}{0}{ssfletters}{'002}
\DeclareMathSymbol{\bsfLambda}{0}{bsfletters}{'003}
\DeclareMathSymbol{\ssfLambda}{0}{ssfletters}{'003}
\DeclareMathSymbol{\bsfXi}{0}{bsfletters}{'004}
\DeclareMathSymbol{\ssfXi}{0}{ssfletters}{'004}
\DeclareMathSymbol{\bsfPi}{0}{bsfletters}{'005}
\DeclareMathSymbol{\ssfPi}{0}{ssfletters}{'005}
\DeclareMathSymbol{\bsfSigma}{0}{bsfletters}{'006}
\DeclareMathSymbol{\ssfSigma}{0}{ssfletters}{'006}
\DeclareMathSymbol{\bsfUpsilon}{0}{bsfletters}{'007}
\DeclareMathSymbol{\ssfUpsilon}{0}{ssfletters}{'007}
\DeclareMathSymbol{\bsfPhi}{0}{bsfletters}{'010}
\DeclareMathSymbol{\ssfPhi}{0}{ssfletters}{'010}
\DeclareMathSymbol{\bsfPsi}{0}{bsfletters}{'011}
\DeclareMathSymbol{\ssfPsi}{0}{ssfletters}{'011}
\DeclareMathSymbol{\bsfOmega}{0}{bsfletters}{'012}
\DeclareMathSymbol{\ssfOmega}{0}{ssfletters}{'012}

\newcommand{\tilp}{\tilde{p}}

\newcommand{\tilX}{\tilde{X}}

\newcommand{\eps}{\varepsilon}

\DeclareMathOperator*{\argmin}{arg\,min}

\newcommand{\bone}{\mathbf{1}}

\newtheorem{theorem}{Theorem} 
\newtheorem{lemma}[theorem]{Lemma}

\newtheorem{remark}{Remark}

\newcommand{\qednew}{\nobreak \ifvmode \relax \else
      \ifdim\lastskip<1.5em \hskip-\lastskip
      \hskip1.5em plus0em minus0.5em \fi \nobreak
      \vrule height0.75em width0.5em depth0.25em\fi}

%% file: ge_prf.tex
\section{Modified G\"artner-Ellis theorem}\label{app:ge}
Here we present and prove a modified form of the G\"artner-Ellis theorem with a shift and a general order (normalization).

 Some of the ideas (for example  the proof of \eqref{ge1_upper} for the case $\beta_n\gamma_n\to\infty$ and \eqref{ge2}) are contained in \cite[Theorem 2.3.6]{Dembo} but other elements of the proof are novel. To keep the exposition self-contained, we provide all the details of the proof for the event of interest  $\{ X_n\le x\beta_n\}$. The  standard G\"artner-Ellis theorem \cite[Theorem 2.3.6]{Dembo}  applies in full generality to open and closed sets but here we are only interested in events of the form $\{ X_n\le x\beta_n\}$.
\begin{theorem}[Modified G\"artner-Ellis theorem]\label{thm7}
We consider three sequences $\alpha_n,\beta_n,\gamma_n$. The sequence $\alpha_n$ is arbitrary and $\beta_n$ and $\gamma_n$ are positive sequences that additionally  satisfy
$\beta_n \to \infty$   and $\gamma_n \to 0$.
Let $p_n$ be a sequence of distributions,
and $X_n$ be the sequence of random variables with distribution $p_n$.
Define the cumulant generating function 
\begin{equation}
\mu_n(\theta):=\log\bbE_{p_n}[ \exp (\theta X_n)].\label{eqn:def_mu_n}
\end{equation}
Let $\theta_0\le 0$ and $\nu_1$ be constants. Assume that 
\begin{equation}
\mu_n(\theta_0+ \gamma_n y)
= \alpha_n+ \beta_n \nu_1+\beta_n \gamma_n \nu_{2,n}(y) 
\label{eqn:mu_n}
\end{equation}
for some sequence of functions $ \nu_{2,n}$. 
Assume that 
\begin{equation}
\nu_{2,n}(y)\to \nu_2(y)\quad\mbox{pointwise,} \label{eqn:ptwise}
\end{equation}
and the limiting function $\nu_2(y)$ satisfies 
\begin{enumerate}
\item $\nu_2(0)=0$;
\item $\nu_2(y)$ is strictly convex;
\item $\nu_2(y)$ is $C^2$-continuous 
on an open subset $G\subset \mathbb{R}$
\end{enumerate}
We also fix $x$ and the real number $y_0 \in G$ satisfying\footnote{Such a real number $y_0$ is guaranteed to exist because $\nu_2$ is strictly convex so $\nu_2'$ is strictly increasing.}
\begin{equation}
\nu_2'(y_0)=x.  \label{eqn:nu2_pr}
\end{equation}

\begin{itemize}
\item[(i)] 
When $\theta_0 <0$, we consider two subcases:
\begin{itemize}
\item[(a)] $\beta_n \gamma_n\to\infty$;
\item[(b)] $\beta_n=\gamma_n^{-1}$ and $\nu_2$ is  a quadratic function.\footnote{Since $\nu_2$ is strictly convex and quadratic in this case, it must be of the form $\nu_2(y) = \varrho_0 + \varrho_1 y+\varrho_2y^2$ for some constants $\varrho_2>0$ and $\varrho_0,\varrho_1\in\bbR$. }
\end{itemize}
In both cases, we have the lower bound
\begin{align}
- \log p_n \bigg\{ \frac{X_n}{\beta_n} \le x\bigg\}\ge  -\alpha_n+ (\theta_0 x- \nu_1)\beta_n+(y_0 x - \nu_2(y_0) ) \beta_n \gamma_n +   o(\beta_n\gamma_n) \label{ge1_lower}
\end{align}
and the upper bound
\begin{equation}
- \log p_n \bigg\{ \frac{X_n}{\beta_n} \le x\bigg\} \le   -\alpha_n+ (\theta_0 x- \nu_1)\beta_n+o(\beta_n). \label{ge1_upper}
\end{equation}
\item[(ii)]
When $\theta_0 =0$, $y_0 <0$, $\alpha_n=\nu_1=0$,
and
$\beta_n \gamma_n \to \infty$,  
\begin{align}
- \log p_n \bigg\{ \frac{X_n}{\beta_n} \le x\bigg\}
=
(y_0 x - \nu_2(y_0) ) \beta_n \gamma_n
+  o(\beta_n \gamma_n).
\label{ge2}
\end{align}
\end{itemize}
Note that in Case (i), $y_0$ may take any real value.
However, 
in   Case (ii), $y_0$ always takes a negative value.
\end{theorem}


The proofs of these statements are split into four parts. In Appendix \ref{sec:prf_lower}, we prove the lower bounds to all statements. In Appendix \ref{sec:upper_inf}, we prove the upper bound for Case (ia) with $\beta_n\gamma_n\to\infty$. In Appendix \ref{sec:upper_consta}, we prove the upper bound for  Case (ib) with $\beta_n=\gamma_n^{-1}$ and $\nu_2$ is a quadratic. In Appendix \ref{sec:prf_upper_caseii}, we prove the upper bound for Case (ii). 

\subsection{Proof of Lower Bounds of Both Cases in \eqref{ge1_lower} and \eqref{ge2}} \label{sec:prf_lower}
\begin{proof} The proofs of the lower bounds all cases are common. 

First note that for Case (i),  $\theta_0<0$  and $\gamma_n\to 0$, so for sufficiently large $n$, we have $\theta_0 + \gamma_n y_0<0$. For Case (ii), even though $\theta_0=0$, $y_0<0$ so similarly, we have $\theta_0 + \gamma_n y_0<0$. 
Thus, using  Markov's inequality,
\begin{align}
p_n \bigg\{ \frac{X_n}{\beta_n} \le x\bigg\}  &= p_n \bigg\{ \exp\left[ \Big(\frac{X_n}{\beta_n} -x \Big)\beta_n (\theta_0 + \gamma_n y_0)  \right]\ge 1 \bigg\}   \\
&\le \bbE_{p_n}\bigg\{ \exp\left[ \Big(\frac{X_n}{\beta_n} -x \Big)\beta_n (\theta_0 + \gamma_n y_0)  \right]\bigg\} .
\end{align}
In other words, 
\begin{align}
- \log p_n \bigg\{ \frac{X_n}{\beta_n} \le x\bigg\} 
&\ge 
x \beta_n \theta_0 +x\gamma_n\beta_n y_0-\mu_n(\theta_0 +\gamma_n y_0) \\
&=-\alpha_n + (\theta_0 x - \nu_1) \beta_n + (y_0 x - \nu_{2,n}(y_0) ) \beta_n\gamma_n \label{eqn:use_expansion}
\end{align}
where \eqref{eqn:use_expansion} follows from the expansion  of $\mu_n(\theta_0+\gamma_n y)$ in \eqref{eqn:mu_n}. 
So from the assumption that $\nu_{2,n}(y_0)\to  \nu_{2}(y_0)$ (cf.\ \eqref{eqn:ptwise}), we obtain
\begin{align}
- \log p_n \bigg\{ \frac{X_n}{\beta_n} \le x\bigg\} 
&\ge  -\alpha_n + (\theta_0 x - \nu_1) \beta_n + (y_0 x - \nu_{2 }(y_0) ) \beta_n\gamma_n +o (\beta_n\gamma_n)\label{eqn:liminf_bd}
\end{align}
as desired.  This completes the proofs of the lower bounds in \eqref{ge1_lower} and \eqref{ge2}. \end{proof}

\subsection{Proof of Upper Bound of Case (ia)  in \eqref{ge1_upper}, i.e.,  $\beta_n\gamma_n\to\infty$ } \label{sec:upper_inf}
\begin{proof}
The proof of this upper bound proceeds in three distinct steps. 

{\em Step 1 (Measure Tilting)}: First fix a constant $\delta>0$. For the sake of brevity, define the     set
\begin{equation}
\calD_{x,\delta}:=\left\{ \omega: x-2\delta\le\frac{X_n(\omega)}{\beta_n} \le x \right\}.
\end{equation}
It suffices to lower bound the $p_{n}$-probability of the set $\calD_{x,\delta}$ because $\{\omega: X_n(\omega)/\beta_n\le x\}\supset \calD_{x,\delta}$. 
Next given the constant $\delta>0$, we can define the point 
\begin{equation}
x ':=x-\delta.  \label{eqn:def_xp}
\end{equation}
Correspondingly, also define the point $y'$ such that 
\begin{equation}
\nu_2'(y') =x'.  \label{eqn:prop_nu2}
\end{equation}
Let $\theta$ be defined as 
\begin{equation}
\theta:=\theta_0+\gamma_n y', \label{eqn:def_the}
\end{equation}
where $\theta_0$ and $\gamma_n$ are fixed in the theorem statement. 
 Define  the tilted probability measure
\begin{equation}
\tilp_{n,\theta}(\omega):=p_{n}(\omega) \exp(\theta X_n(\omega)- \mu_n(\theta))  .\label{eqn:tilted}
\end{equation}
Now, for all $n$ large enough $\theta<0$ since $\theta_0<0$ and $\gamma_n\to 0$. Thus, from the definition of the tilted measure $\tilp_{n,\theta}(\omega)$, for those $\omega\in\calD_{x,\delta}$, we have 
\begin{align}
\tilp_{n,\theta } (\omega) &\le p_n(\omega) \exp \left( \theta \beta_n (x-2\delta) -\mu_n(\theta )\right) .
\end{align}
Integrating over all $\omega \in\calD_{x,\delta}$, taking the logarithm and normalizing by $\beta_n $ we obtain
\begin{equation}
  \frac{1}{\beta_n }\log  \tilp_{n,\theta}\big\{ \calD_{x,\delta}\big\}  \le\frac{1}{\beta_n} \log p_n\big\{ \calD_{x,\delta} \big\}  +  \theta ( x-2\delta) - \frac{\mu_n(\theta)}{\beta_n }.
 \end{equation} 
 Substituting the definition of $\theta$ in \eqref{eqn:def_the} into the above, we obtain
 \begin{equation}
  \frac{1}{\beta_n}\log  \tilp_{n,\theta_0+\gamma_n y'}\big\{ \calD_{x,\delta}\big\}  \le\frac{1}{\beta_n} \log p_n\big\{ \calD_{x,\delta} \big\}  +  (\theta_0+\gamma_n y') (x-2\delta)- \frac{\mu_n(\theta_0+\gamma_n y')}{\beta_n}.
 \end{equation}
Using the expansion of $\mu_n(\cdot)$ in \eqref{eqn:mu_n} in the above, we obtain
\begin{equation}
\frac{1}{\beta_n}\log  \tilp_{n,\theta_0+\gamma_n y'}\big\{ \calD_{x,\delta}\big\}  \le\frac{1}{\beta_n} \log p_n\big\{ \calD_{x,\delta} \big\}  + (\theta_0+\gamma_n y') (x-2\delta)- \frac{\alpha_n +\beta_n \nu_1 +\beta_n\gamma_n\nu_{2,n}(y')}{\beta_n}.
\end{equation}
Rearranging, we obtain
 \begin{align}
& \frac{1}{\beta_n}\left[ -\log p_n\big\{ \calD_{x,\delta} \big\} + \alpha_n-\beta_n(\theta_0 x-\nu_1)\right]\nn\\*
 &\qquad\le  -2\delta (\theta_0 +\gamma_n y')+\gamma_n\big( xy'-\nu_{2,n}(y')\big) - \frac{1}{\beta_n }\log \tilp_{n, \theta_0+\gamma_n y'} \big\{ \calD_{x,\delta}\big\} . \label{eqn:shift_alpha} 
 \end{align}
 
 {\em Step 2 (Bounding the Probability in \eqref{eqn:shift_alpha})}: 
 Now, our aim is to lower bound the probability in the final term   in \eqref{eqn:shift_alpha} namely $\tilp_{n, \theta_0+\gamma_n y'} \big\{ \calD_{x,\delta}\big\}$. To this end, define the cumulant generating function with respect to the tilted measure $\tilp_{n,\theta_0+\gamma_n y'}$ as follows:
 \begin{equation}
 \tilde{\mu}_n(\lambda):=\log\bbE_{\tilp_{n,\theta_0+\gamma_n y'}}[\exp(\lambda X_n)].
 \end{equation}
 Now observe that for any $s\in\bbR$, 
 \begin{align}
 \tilde{\mu}_n(\gamma_n s) & = \log \int_\bbR \exp\big(\gamma_n s X(\omega) \big) \, \tilp_{n,\theta_0+\gamma_n y'}(\rmd\omega) \\
&=\log \int_\bbR \exp\big(\gamma_n s X(\omega)\big)\exp\big( (\theta_0 + \gamma_n y') X(\omega) \big) \exp\big(-\mu_n ( \theta_0 + \gamma_n y')\big)  \, p_{n}(\rmd\omega) \label{eqn:relate_tilt}\\
&= \left[\log \int_\bbR \exp\big(\gamma_n s X(\omega) \big) \exp\big( (\theta_0 + \gamma_n y') X(\omega) \big)  \, p_{n}(\rmd\omega) \right] - \mu_n ( \theta_0 + \gamma_n y')\label{eqn:relate_tilt2}\\
&=\mu_n ( \theta_0 + \gamma_n (s+y'))-\mu_n ( \theta_0 + \gamma_n y')\\
 &=\beta_n \gamma_n \big[ \nu_{2,n}(s+y')-\nu_{2,n}(y') \big], \label{eqn:use_mu_n}
 \end{align}
 where \eqref{eqn:relate_tilt} is due to \eqref{eqn:tilted}, and \eqref{eqn:use_mu_n} is due to \eqref{eqn:mu_n}. Thus, by normalizing by $\beta_n\gamma_n\to\infty$, and noting that $\nu_{2,n}$ converges to $\nu_2$ pointwise (cf.\ \eqref{eqn:ptwise}), we have 
 \begin{equation}
 \lim_{n\to\infty} \frac{1}{\beta_n\gamma_n} \tilde{\mu}_n(\gamma_n s)=\nu_{2}(s+y')-\nu_{2}(y').\label{eqn:xi_def0}
 \end{equation}
Thus, from this calculation, we can conclude that 
\begin{equation}
\xi (s;y'):=\nu_{2}(s+y')-\nu_{2}(y') \label{eqn:xi_def}
\end{equation}
as a function of $s$, is the limiting cumulant generating function of the sequence of measures $\tilp_{n,\theta_0+\gamma_n y'}$.  By the union bound, 
\begin{equation}
\tilp_{n,\theta_0+\gamma_n y'} \big\{ \calD_{x,\delta}^c \big\} \le  \tilp_{n,\theta_0+\gamma_n y'} \big\{ \calE_1\big\} + \tilp_{n,\theta_0+\gamma_n y'} \big\{ \calE_2 \big\} 
\end{equation}
where 
\begin{align}
\calE_1 &:= \left\{ \omega: \frac{X_n(\omega)}{\beta_n} < x-2\delta \right\},\quad\mbox{and}\\
\calE_2 &:= \left\{ \omega: \frac{X_n(\omega)}{\beta_n} > x  \right\}.
\end{align}
We analyze the $\tilp_{n,\theta_0+\gamma_n y'} $-probability of $\calE_2$ first. 
By  Markov's inequality, for any fixed $s\ge 0$, 
\begin{align}
\tilp_{n,\theta_0+\gamma_n y'} \big\{ \calE_2 \big\} &=\tilp_{n,\theta_0+\gamma_n y'} \big\{X_n >\beta_n x\big\} \\
&=\tilp_{n,\theta_0+\gamma_n y'} \big\{ \exp(s\gamma_n X_n )>\exp(s\beta_n \gamma_nx)\big\} \\
&\le \frac{ \bbE_{ \tilp_{n,\theta_0+\gamma_n y'}} [ \exp(s \gamma_n X_n)] }{\exp(s \beta_n \gamma_n x)}.
\end{align}
Taking logarithms, normalizing by $\beta_n\gamma_n \to\infty$ and taking the $\limsup$, we obtain
\begin{align}
\limsup_{n\to\infty}\frac{1}{\beta_n\gamma_n} \log\tilp_{n,\theta_0+\gamma_n y'} \big\{ \calE_2 \big\}  \le  - sx + \limsup_{n\to\infty} \frac{1}{\beta_n\gamma_n} \log \bbE_{\tilp_{n,\theta_0+\gamma_n y'}} [ \exp(s\gamma_nX_n) ]
\end{align}
Now using the definition of $s\mapsto\xi(s; y')$ in  \eqref{eqn:xi_def0}--\eqref{eqn:xi_def}, we conclude that
\begin{equation}
\limsup_{n\to\infty}\frac{1}{\beta_n\gamma_n} \log\tilp_{n,\theta_0+\gamma_n y'} \big\{\calE_2\big\}  \le  -sx + \xi(s;y').
\end{equation}
Since $s\ge 0$ is arbitrary, we have 
\begin{equation}
\limsup_{n\to\infty}\frac{1}{\beta_n\gamma_n} \log\tilp_{n,\theta_0+\gamma_n y'} \big\{ \calE_2 \big\} \le -  \sup_{s\ge 0} \left\{ sx-\xi(s;y')\right\} . \label{eqn:optimized_s}
\end{equation}
Define the Fenchel-Legendre transform 
\begin{equation}
\xi^*(x;y'):= \sup_{s\ge 0} \left\{ sx-\xi(s;y')\right\}. \label{eqn:fl_trans}
\end{equation}
Now, we claim that   $\xi^*(x;y')>0$. Let $s_x^*$ achieve the supremum in \eqref{eqn:fl_trans}. Consider the following optimality condition for the convex optimization problem in \eqref{eqn:fl_trans}:
\begin{equation}
x- \nu_2'(s_x^*+y') = 0  . \label{eqn:optcon}
\end{equation}
Since $\nu_2$ is  assumed to be strictly convex, so $\nu_2'$ is strictly increasing. Furthermore, it has the property (cf.\ \eqref{eqn:prop_nu2}) that $\nu_2'(y') = x'$. Since $x'<x$, by continuity of $\nu_2'$, the optimal $s_x^*$ in \eqref{eqn:optcon} is positive. Since the strict convexity of $\nu_2$ means the same for $s\mapsto\xi( s;y')$, this implies that   $\xi^*(x;y')$   is positive.   So we conclude that 
\begin{equation}
\tilp_{n,\theta_0+\gamma_n y'} \big\{ \calE_2\big\}\le \exp(-\beta_n\gamma_n \tau_2) \label{eqn:converge_E2}
\end{equation}
for some $\tau_2>0$ and for all $n$ large enough. In a completely analogous way, we can show that 
\begin{align}
\limsup_{n\to\infty}\frac{1}{\beta_n\gamma_n} \log\tilp_{n,\theta_0+\gamma_n y'} \big\{ \calE_1 \big\} \le -  \sup_{s\le 0} \left\{ s(x-2\delta) -\xi(s;y')\right\} .
\end{align}
Defining
\begin{equation}
\tilde{\xi}^*(x-2\delta;y'):= \sup_{s\le 0} \left\{ s(x-2\delta)-\xi(s;y')\right\},\label{eqn:fl_trans_minus}
\end{equation}
and examining the optimality condition for $s$ in \eqref{eqn:fl_trans_minus}, we see that $\tilde{\xi}^*(x-2\delta;y')>0$ and so 
\begin{equation}
\tilp_{n,\theta_0+\gamma_n y'} \big\{ \calE_1\big\}\le \exp(-\beta_n\gamma_n \tau_1)  \label{eqn:converge_E1}
\end{equation}
for some $\tau_1>0$ and for all $n$ large enough. 
Consequently, we have 
\begin{equation}
\tilp_{n,\theta_0+\gamma_n y'} \big\{ \calD_{x,\delta}\big\}\ge 1- 2\exp(-\beta_n\gamma_n \tau ) \label{eqn:converge_D}
\end{equation}
where $\tau:=\min\{\tau_1, \tau_2\}>0$.

 {\em Step 3 (Considering Asymptotics)}: 
Now substituting \eqref{eqn:converge_D} back into \eqref{eqn:shift_alpha}, we obtain 
\begin{align}
 &\frac{1}{\beta_n}\left[ -\log p_n\big\{ \calD_{x,\delta} \big\} + \alpha_n-\beta_n(\theta_0 x-\nu_1)\right]\nn\\*
 &\qquad\le  -2\delta (\theta_0 +\gamma_n y') + \gamma_n\big( xy'-\nu_{2,n}(y')\big)- \frac{1}{\beta_n }\log \big(   1-2\exp(-\beta_n\gamma_n \tau)\big) .  \label{eqn:generic}
\end{align} 
Since $\beta_n\gamma_n\to\infty$, the final term vanishes. Consequently,  taking limits, we have 
\begin{equation}
\limsup_{n\to\infty} \frac{1}{\beta_n }\left[ -\log p_n\big\{ \calD_{x,\delta} \big\} + \alpha_n-\beta_n(\theta_0 x-\nu_1)\right]\le  -2\delta  \theta_0  . \label{eqn:delta_rem}
\end{equation}
Since $\delta>0$ is arbitrary, we may take $\delta\downarrow 0$ (also recall that $\theta_0<0$ so the term on the right-hand-side of \eqref{eqn:delta_rem}  is non-negative) to conclude that 
\begin{equation}
\limsup_{n\to\infty} \frac{1}{\beta_n }\left[ -\log p_n\big\{ \calD_{x,\delta} \big\} + \alpha_n-\beta_n(\theta_0 x-\nu_1)\right]\le  0 .
\end{equation}
This concludes the proof of the upper bound of \eqref{ge1_upper} for the Case (ib), i.e., $\beta_n\gamma_n\to\infty$.
\end{proof}

\subsection{Proof of Upper Bound of Case (ib)  in \eqref{ge1_upper}, i.e.,  $\beta_n=\gamma_n^{-1}$ } \label{sec:upper_consta}
\begin{proof}
Recall that $\theta_0<0$ and $\beta_n=\gamma_n^{-1}$ consequently $\beta_n\gamma_n=1$.  Define the tilted probability measure
\begin{align}
\tilp_n(\omega):=
p_{n}(\omega)\exp\big( (\theta_0+\gamma_n y_0)X_n (\omega)-\mu_n(\theta_0+\gamma_n y_0) \big). \label{eqn:tiltd}
\end{align}
Fix $s\in\bbR$.  Note from the strict convexity of $\nu_2$ that $\nu_2''(y_0)>0$ for all $y_0$.  Then, using $\bbE_{\tilp_n}[\cdot]$ to denote the expectation with respect to the distribution $\tilp_n$ in~\eqref{eqn:tiltd},  we have
\begin{align}
& \log 
\bbE_{\tilp_n}\left[\exp\bigg( \sqrt{\frac{\gamma_n}{\beta_n \nu_2''(y_0)}}s (X_n - x \beta_n ) \bigg) \right]
\nn\\
&=\log\int_\bbR \exp\bigg( \sqrt{\frac{\gamma_n}{\beta_n \nu_2''(y_0)}}s (X_n(\omega) - x \beta_n ) \bigg)\,\rmd \tilp_n(\omega) \label{eqn:chg_meas0} \\
&=\log\int_\bbR \exp\bigg( \sqrt{\frac{\gamma_n}{\beta_n \nu_2''(y_0)}}s (X_n(\omega) - x \beta_n ) \bigg)\exp\big( (\theta_0+\gamma_n y_0)X_n (\omega)-\mu_n(\theta_0+\gamma_n y_0) \big) \,\rmd p_{n }(\omega) \label{eqn:chg_meas} \\
&=
\mu_n \left(\theta_0+\gamma_n \bigg(y_0 +\sqrt{\frac{1}{\gamma_n \beta_n \nu_2''(y_0)}}s\bigg) \right)
-\mu_n(\theta_0+\gamma_n y_0)
-\sqrt{\frac{\gamma_n}{\beta_n \nu_2''(y_0)}} s x \beta_n \label{eqn:use_def_mu}\\
&=\nu_{2,n} \bigg(y_0 +\sqrt{\frac{1}{ \nu_2''(y_0)}}s\bigg)
- \nu_{2,n} (y_0)
-\sqrt{\frac{ 1}{\nu_2''(y_0)}} s x  \label{eqn:use_expansion_mu}\\
&= \nu_2  \bigg(y_0 +\sqrt{\frac{1}{  \nu_2''(y_0)}}s\bigg) - \nu_{2 } (y_0)+ o(1)  
-\sqrt{\frac{1 }{\nu_2''(y_0)}} s x   \label{eqn:converge_nu2} 
\end{align}
where 
\begin{enumerate}
\item \eqref{eqn:chg_meas} follows from the change of measure per \eqref{eqn:tiltd};
\item \eqref{eqn:use_def_mu} follows from the definition of  $\mu_n(\cdot)$ in \eqref{eqn:def_mu_n};
\item \eqref{eqn:use_expansion_mu} follows from the expansion of $\mu_n(\theta_0+ \gamma_n y)$ in \eqref{eqn:mu_n} and the fact that $\beta_n\gamma_n=1$;
  \item and \eqref{eqn:converge_nu2} follows from the pointwise convergence of $\nu_{2,n}$ to $\nu_2$ in \eqref{eqn:ptwise}.
\end{enumerate}
Hence, taking limits and noting that $x= \nu_2'(y_0)$ (cf.\ \eqref{eqn:nu2_pr}),  we obtain
\begin{equation}
\lim_{n\to\infty} \log 
\bbE_{\tilp_n}\left[\exp\bigg( \sqrt{\frac{\gamma_n}{\beta_n \nu_2''(y_0)}}s (X_n - x \beta_n ) \bigg) \right]= \nu_2  \bigg(y_0 +\sqrt{\frac{1}{  \nu_2''(y_0)}}s\bigg) - \nu_{2 } (y_0) 
- s \nu_2'(y_0)\sqrt{\frac{1 }{\nu_2''(y_0)}}.  
\end{equation}
Since we assumed that $\nu_2$ is  a quadratic function, the second-order Taylor approximation of $\nu_2$ at $y_0$ is exactly a quadratic so
\begin{equation} \label{eqn:conv_gauss}
\lim_{n\to\infty} \log 
\bbE_{\tilp_n}\left[\exp\bigg( \sqrt{\frac{\gamma_n}{\beta_n \nu_2''(y_0)}}s (X_n - x \beta_n ) \bigg) \right]= \frac{1}{2}s^2.
\end{equation}
The relation in \eqref{eqn:conv_gauss} says that the sequence of random variables $\sqrt{\frac{\gamma_n}{\beta_n \nu_2''(y_0)}} (X_n - x \beta_n )$ with corresponding distribution $\tilp_n$ has a sequence of  cumulant generating functions that converges pointwise  to the   quadratic function $\frac{1}{2}s^2$. Thus,  we can conclude that it converges in distribution to the standard Gaussian by L\'evi's continuity theorem \cite[Thm.~18.21]{Fristedt}.

Furthermore,  from~\eqref{eqn:tiltd}, we have 
\begin{align}
& p_n \left\{ X_n \le x {\beta_n}\right\}
\exp\big( -\mu_n(\theta_0+\gamma_n y_0) \big)
\exp\big(  (\theta_0 + \gamma_n y_0 )x \beta_n \big) \nn\\
&= 
\int_{\Omega}
\exp\big( - (\theta_0 + \gamma_n y_0 ) (X_n (\omega) - x \beta_n )\big)\bone\left\{ X_n (\omega) \le x {\beta_n} \right\} \, \rmd \tilp_n(\omega). \label{eqn:p_n_int}
\end{align}
Now, for notational brevity, define 
\begin{equation}
a_n  :=-(\theta_0 +\gamma_n y_0)\sqrt{ \frac{\beta_n\nu_2''(y_0) }{\gamma_n}}  .\label{eqn:defa}
\end{equation}
Because $\theta_0<0$, $\gamma_n\to 0$ and $\beta_n\to\infty$, we conclude that $a_n\ge 0$ for $n$ large enough and also $a_n\to \infty$. 

Let the probability measure corresponding to the random variable $\sqrt{\frac{\gamma_n }{\beta_n \nu_2''(y_0)}} (X_n - x \beta_n )$ with respect to the distribution $\tilp_n$  in \eqref{eqn:tiltd} be denoted as  $\bbP_n$. More precisely, for every Borel measurable set $\calE$, we have  the relation
\begin{equation}
\bbP_n(\calE) :=  \int_\Omega\, \bone\left\{ \sqrt{\frac{\gamma_n}{\beta_n \nu_2''(y_0)}} (X_n (\omega) - x \beta_n ) \in\calE \right\}  \, \rmd \tilp_n( \omega)   .
\end{equation}
This is the relation between the measures $\tilp_n$ and $\bbP_n$, and it is this change-of-measure step (from $\tilp_n$ to $\bbP_n$) that is crucial in this proof. 
 Note from the calculation leading to \eqref{eqn:conv_gauss} that $\bbP_n $ converges weakly to the standard Gaussian measure $\bbP(\calA) :=\int_{\calA}\varphi(w)\,\rmd w$. Thus, through a change of variables
\begin{equation}
  \omega\mapsto z = \sqrt{ \frac{\gamma_n}{\beta_n\nu_2''(y_0 )}}  (X_n(\omega)-x\beta_n ),
 \end{equation} 
 the quantity in \eqref{eqn:p_n_int} can be expressed as  the   integral
\begin{align}
\Psi_n := \int_{-\infty}^0 \exp(a_n z) \, \rmd\bbP_n( z).
\end{align}
We now provide upper and lower bounds for this integral to understand its asymptotic behavior. 
We have 
\begin{align}
\Psi_n  \le \int_{-\infty}^{0}\, \rmd\bbP_n( z) \le \frac{3}{4}\label{eqn:3quarter}  
\end{align}
where the last bound follows for all $n$ sufficiently large due to L\'evi's continuity theorem \cite[Thm.~18.21]{Fristedt}, i.e., $\bbP_n\to\bbP$ in distribution where $\bbP$ is a standard Gaussian distribution. Obviously, $\bbP( (\infty, 0]) =\frac{1}{2}$. 

To lower bound $\Psi_n$, we fix $\eps>0$ consider
\begin{align}
\Psi_n  &\ge  \int_{-\infty}^{- {\eps} }  \exp(a_n z) \, \rmd\bbP_n( z) \\
&\ge   \exp(-\eps  a_n) \int_{-\infty}^{- {\eps} }  \, \rmd\bbP_n( z) \label{eqn:lower_limt}  \\
&\ge   \frac{1}{4} \exp(-\eps  a_n)  \label{eqn:quarter}  
\end{align}
where  in \eqref{eqn:lower_limt} we substituted the upper limit into the integrand and in \eqref{eqn:quarter}  we assumed $\eps$ is chosen small enough so that $\bbP_n( (-\infty,-\eps])\ge\frac{1}{  4}$ for all $n$ large enough. 
In sum, \eqref{eqn:3quarter} and \eqref{eqn:quarter} yield
\begin{equation}
\Psi_n = O(1) \cap  \Omega\big( \exp(- \eps a_n) \big). \label{eqn:asympI}
\end{equation}
Recall from \eqref{eqn:p_n_int} that
\begin{equation}
 p_n \left\{ X_n \le x {\beta_n}\right\}
\exp\big( -\mu_n(\theta_0+\gamma_n y_0) \big)
\exp\big(  (\theta_0 + \gamma_n y_0 )x \beta_n \big) = \Psi_n.
\end{equation}
Therefore,
\begin{align}
 -\log  p_n \left\{ X_n \le x {\beta_n}\right\}  
&= 
-\mu_n(\theta_0+\gamma_n y_0)
+(\theta_0 + \gamma_n y_0 )x \beta_n 
 - \log \Psi_n\\
 &=
-\alpha_n -\beta_n \nu_1 - \beta_n \gamma_n \nu_{2,n}(y_0)
+(\theta_0 + \gamma_n y_0 )x \beta_n 
 - \log \Psi_n\\
 &=
-\alpha_n 
+\beta_n (\theta_0 x -\nu_1) 
+\beta_n \gamma_n (y_0 x -\nu_{2}(y_0) +o(1))
 - \log \Psi_n, \label{eqn:nu_conv}\\
  &=
-\alpha_n 
+\beta_n (\theta_0 x -\nu_1) 
+O(1) 
 - \log \Psi_n, \label{eqn:nu_conv2}
 \end{align}
 where  \eqref{eqn:nu_conv} holds from the pointwise convergence of $\nu_{2,n}$ to $\nu_2$ (cf.\ \eqref{eqn:ptwise}), and \eqref{eqn:nu_conv2} holds because $\beta_n\gamma_n=1$ and $y_0 x -\nu_{2}(y_0) = O(1)$.    Now using the   the asymptotic behavior of $\Psi_n$ in  \eqref{eqn:asympI}, we obtain
 \begin{equation}
-\kappa \le\left[ -\log  p_n \left\{ X_n \le x {\beta_n}\right\}   +\alpha_n-\beta_n (\theta_0 x -\nu_1) \right]   +O(1) \le \kappa\eps a_n
 \end{equation}
 for some  finite constant $\kappa> 0$.   Since  $\theta_0<0$, from the definition of $a_n$ in \eqref{eqn:defa}, we know that 
  $a_n$ is of the order  $O(\sqrt{\beta_n\gamma_n^{-1}})$. Additionally,  since $\eps>0$ is arbitrarily small,
 \begin{equation}
-\kappa\le\left[ -\log  p_n \left\{ X_n \le x {\beta_n}\right\}   +\alpha_n-\beta_n (\theta_0 x -\nu_1) \right]  +O(1)\le o\Big(\sqrt{\beta_n\gamma_n^{-1}} \Big). \label{eqn:order_a}
 \end{equation} 
 Now, recall that $\beta_n=\gamma_n^{-1}$. So $o\big(\sqrt{\beta_n\gamma_n^{-1}}\big)=o(\beta_n)$ and we have finished the proof of   the upper bound in \eqref{ge1_upper} for the case $\beta_n=\gamma_n^{-1}$  and $\nu_2$ being a quadratic. 
 \end{proof}
 
 \subsection{Proof of Upper Bound of Case (ii)  in \eqref{ge2} } \label{sec:prf_upper_caseii}
 
\begin{proof}
The proof of this case proceeds similarly to that in Appendix \ref{sec:upper_inf}. We only highlight the main differences here. The measure tilting step proceeds similarly with the exception that  $\theta=\gamma_n y'$ (cf.\ \eqref{eqn:def_the}). Because $\delta$ can be chosen arbitrarily small, $y'<0$ by continuity since $y_0<0$. Thus similarly to the proof in Appendix \ref{sec:upper_inf}, $\theta$ is a negative sequence.

Since $\theta_0=0$, \eqref{eqn:generic} reduces to 
\begin{equation}
 \frac{1}{\beta_n\gamma_n}\left[ -\log p_n\big\{ \calD_{x,\delta} \big\} + \alpha_n-\beta_n(\theta_0 x-\nu_1)\right]\le -2\delta y'  +  \big( xy'-\nu_{2,n}(y')\big) - \frac{1}{\beta_n \gamma_n}\log \big(   1-2\exp(-\beta_n\gamma_n b)\big) .  \label{eqn:generic2}
\end{equation}  
Since $\beta_n\gamma_n\to\infty$, the final term vanishes when we take limits, yielding
\begin{equation}
\limsup_{n\to\infty} \frac{1}{\beta_n\gamma_n}\left[ -\log p_n\big\{ \calD_{x,\delta} \big\} + \alpha_n-\beta_n(\theta_0 x-\nu_1)\right]\le -2\delta y'  +  \big( xy'-\nu_{2}(y')\big) .
\end{equation}
Now take $\delta\downarrow 0$, we obtain by the continuity of $\nu_2'$ that $y'\to y_0$. Thus,  we have
\begin{equation}
\limsup_{n\to\infty} \frac{1}{\beta_n\gamma_n}\left[ -\log p_n\big\{ \calD_{x,\delta} \big\} + \alpha_n-\beta_n(\theta_0 x-\nu_1)\right]\le  xy_0-\nu_{2}(y_0)
\end{equation}
as desired.
\end{proof}
\begin{remark} 
Observe from the above proof that we used two different techniques for the cases $\beta_n\gamma_n\to\infty$ (Appendices \ref{sec:upper_inf} and \ref{sec:prf_upper_caseii}) and $\beta_n=\gamma_n^{-1}$ (Appendix \ref{sec:upper_consta}). The former case follows essentially from the same steps as in the standard proof of the G\"artner-Ellis theorem in \cite[Theorem 2.3.6]{Dembo} (however, see Remark \ref{rmk:tilting}). The latter case cannot be handled using the technique for $\beta_n\gamma_n\to\infty$ because the final term  in \eqref{eqn:generic} fails to vanish with $\beta_n\gamma_n=\mathrm{const}$. Hence, we develop a novel technique based on the weak convergence of $\sqrt{\frac{\gamma_n}{\beta_n \nu_2''(y_0)}} (X_n - x \beta_n )$ (under the tilted measure $\tilp_n$) to handle the case where $\beta_n\gamma_n=\mathrm{const}$  (with the added assumption that $\nu_2$ is a quadratic). This technique may be of independent interest to other problems in probability theory. We note, though, that this technique based on weak convergence cannot be used to handle the case in which $\beta_n\gamma_n\to\infty$. This is because in this case, a careful examination of the steps from \eqref{eqn:chg_meas0} to \eqref{eqn:conv_gauss}  would show that this technique leads to an approximation of  the cumulant generating function $\log 
\bbE_{\tilp_n}\big[\exp\big( \sqrt{\frac{\gamma_n}{\beta_n \nu_2''(y_0)}}s (X_n - x \beta_n ) \big) \big]$ that is too coarse for our needs.
\end{remark}
\begin{remark}\label{rmk:tilting}
Readers familiar with the standard proof of the G\"artner-Ellis theorem in \cite[Theorem 2.3.6]{Dembo} would notice the subtle difference of the   proof  in Appendix \ref{sec:upper_inf} vis-\`a-vis the standard one. The titled measure is $\tilp_{n, \theta_0+\gamma_n y'}$ and $y'$ is defined in terms of $x'$ in \eqref{eqn:def_xp}. In contrast, in \cite[Theorem 2.3.6]{Dembo}, the tilting parameter is chosen to be a fixed exposed hyperplane \cite[Definition 2.3.3]{Dembo}  of the analogue of $x'$. Our tilting parameter, and hence also the tilting distribution, is allowed to vary with $n$.
\end{remark}